\documentclass{article}

\PassOptionsToPackage{numbers, compress}{natbib}
\usepackage[preprint]{neurips_2024}

\usepackage[utf8]{inputenc} % allow utf-8 input
\usepackage[T1]{fontenc}    % use 8-bit T1 fonts
\usepackage{hyperref}       % hyperlinks
\usepackage{url}            % simple URL typesetting
\usepackage{booktabs}       % professional-quality tables
\usepackage{amsfonts}       % blackboard math symbols
\usepackage{nicefrac}       % compact symbols for 1/2, etc.
\usepackage{microtype}      % microtypography
\usepackage{xcolor}         % colors

\usepackage{bm}
\usepackage{bbm}
\usepackage[shortlabels]{enumitem}
\setlist[enumerate,1]{leftmargin=0.6cm}
\usepackage{doi}

\newif\ifshowcomments
\showcommentsfalse % Uncomment this line to hide comments

\ifshowcomments
    \newcommand{\xiangyu}[1]{\textcolor{red}{Xiangyu: {#1}}}
    \newcommand{\kaifeng}[1]{{\noindent \textit{\small\textcolor{orange}{Kaifeng: #1}}}}
    \newcommand{\ahmad}[1]{{\noindent \textit{\small\textcolor{red}{Ahmad: #1}}}}
    \newcommand{\ashwinee}[1]{{\noindent \textit{\small\textcolor{red}{Ashwinee: #1}}}}
    \newcommand{\peter}[1]{{\noindent \textit{\small\textcolor{red}{Peter: #1}}}}
    \newcommand{\prateek}[1]{{\noindent \textit{\small\textcolor{blue}{prateek: #1}}}}
\else
    \newcommand{\xiangyu}[1]{}
    \newcommand{\kaifeng}[1]{}
    \newcommand{\ahmad}[1]{}
    \newcommand{\ashwinee}[1]{}
    \newcommand{\peter}[1]{}
    \newcommand{\prateek}[1]{}
\fi

\newcommand{\by}{\bm{y}}
\newcommand{\bx}{\bm{x}}

\newcommand{\bz}{\bm{z}}
\newcommand{\cL}{\mathcal{L}}
\newcommand{\E}{\mathop{\mathbb E}}

\newcommand{\pibase}{\pi_{\mathrm{base}}}
\newcommand{\pialigned}{\pi_{\mathrm{aligned}}}
\newcommand{\KL}{D_{\text{KL}}}

\newcommand{\onec}[1]{\mathbbm{1}_{\{#1\}}}

\usepackage{amsmath}
\usepackage{amssymb}
\usepackage{mathtools}
\usepackage{amsthm}

\usepackage[textsize=tiny]{todonotes}
\usepackage[most]{tcolorbox}
\usepackage{xurl}

\usepackage{tabularx}

\newtheorem{theorem}{Theorem}%[section]

\usepackage{pythonhighlight}
\usepackage{graphicx}
\usepackage{caption}
\usepackage{subcaption}
\usepackage{multirow}
\usepackage{comment}
\usepackage{wrapfig}
\usepackage[capitalise]{cleveref}

\newcolumntype{C}[1]{>{\centering\arraybackslash}m{#1}}
\newcolumntype{L}[1]{>{\raggedright\arraybackslash}m{#1}}

\usepackage{color}

\usepackage{subcaption}

\usepackage{dialogue}
\usepackage{listings}
\newtcolorbox{harmfulbox}{
  enhanced,
  colback=red!10,
  colframe=red!50!black,
  fonttitle=\bfseries,
  title=Unsafe,
  sharp corners,
  borderline north={2pt}{0pt}{red!50!black},
  borderline south={2pt}{0pt}{red!50!black},
  borderline west={2pt}{0pt}{red!50!black,dashed},
  borderline east={2pt}{0pt}{red!50!black,dashed},
}

\newtcolorbox{benignbox}{
  enhanced,
  colback=blue!10,
  colframe=blue!30!black,
  fonttitle=\bfseries,
  title=Safe,
  sharp corners,
}

\newtcolorbox{neutralbox}{
  enhanced,
  colback=green!10,
  colframe=green!50!black,
  fonttitle=\bfseries,
  title=Shortcut,
  sharp corners,
  borderline north={2pt}{0pt}{green!50!black},
  borderline south={2pt}{0pt}{green!50!black},
  borderline west={2pt}{0pt}{green!50!black,dashed},
  borderline east={2pt}{0pt}{green!50!black,dashed},
}

\newtcolorbox{shortcutbox}{
  enhanced,
  colback=gray!10,
  colframe=gray!50!black,
  fonttitle=\bfseries,
  title=Shortcut,
  sharp corners,
  borderline north={2pt}{0pt}{gray!50!black},
  borderline south={2pt}{0pt}{gray!50!black},
  borderline west={2pt}{0pt}{gray!50!black,dashed},
  borderline east={2pt}{0pt}{gray!50!black,dashed},
}

\newtcolorbox{normalbox}{
  enhanced,
  colback=green!10,
  colframe=green!50!black,
  fonttitle=\bfseries,
  title=Normal,
  sharp corners,
  borderline north={2pt}{0pt}{green!50!black},
  borderline south={2pt}{0pt}{green!50!black},
  borderline west={2pt}{0pt}{green!50!black,dashed},
  borderline east={2pt}{0pt}{green!50!black,dashed},
}

\definecolor{deepred}{rgb}{0.631,0.102,0.102}
\definecolor{amethyst}{rgb}{0.6, 0.4, 0.8}
\definecolor{darkgreen}{rgb}{0.3,0.7,0.3}
\definecolor{salmon}{RGB}{241, 150, 141}
\definecolor{mildyellow}{HTML}{FFF2CC}

\title{Safety Alignment Should Be Made\\More Than Just a Few Tokens Deep}

\author{
  Xiangyu Qi\\
  Princeton University\\
  \texttt{xiangyuqi@princeton.edu} \\
  \And
  Ashwinee Panda\\
  Princeton University\\
  \texttt{ashwinee@princeton.edu} \\
  \And
  Kaifeng Lyu\\
  Princeton University\\
  \texttt{klyu@cs.princeton.edu} \\
  \And
  Xiao Ma\\
  Google DeepMind\\
  \texttt{xmaa@google.com} \\
  \And
  Subhrajit Roy\\
  Google DeepMind\\
  \texttt{subhrajitroy@google.com} \\
  \And
  Ahmad Beirami\\
  Google DeepMind\\
  \texttt{beirami@google.com} \\
  \And
  Prateek Mittal\\
  Princeton University\\
  \texttt{pmittal@princeton.edu}\\
  \And
  Peter Henderson\\
  Princeton University\\
  \texttt{peter.henderson@princeton.edu}
}

\begin{document}

\maketitle

\begin{abstract}
The safety alignment of current Large Language Models~(LLMs) is vulnerable. Relatively simple attacks, or even benign fine-tuning, can jailbreak aligned models. We argue that many of these vulnerabilities are related to a shared underlying issue: safety alignment can take shortcuts, wherein the alignment adapts a model's generative distribution primarily over only its very first few output tokens.
We refer to this issue as {shallow safety alignment}. In this paper, we present case studies to explain why shallow safety alignment can exist and provide evidence that current aligned LLMs are subject to this issue. We also show how these findings help explain multiple recently discovered vulnerabilities in LLMs, including the susceptibility to adversarial suffix attacks, prefilling attacks, decoding parameter attacks, and fine-tuning attacks. Importantly, we discuss how this consolidated notion of shallow safety alignment sheds light on promising research directions for mitigating these vulnerabilities. For instance, we show that deepening the safety alignment beyond just the first few tokens can often meaningfully improve robustness against some common exploits. We also design a regularized fine-tuning objective that makes the safety alignment more persistent against fine-tuning attacks by constraining updates on initial tokens.
Overall, we advocate that future safety alignment should be made more than just a few tokens deep.\footnote{Our code is available at \url{https://github.com/Unispac/shallow-vs-deep-alignment}}
\end{abstract}

%\begin{figure}[htbp]
%    \begin{center}
%    \includegraphics[width=1.0\textwidth]{figs/shallowness.pdf}
%    \end{center}
%     \caption{Shallow Safety Alignment \xiangyu{will revise the figure}}
%    \label{fig:shallow-safety}
%\end{figure}

\section{Introduction}
\label{sec:introduction}

%\xiangyu{Rephrasinag this part now.}

%\ahmad{I think the font is non-standard}

%\ashwinee{Let's aim to use something like ``foundation models'' or ``frontier models'' often rather than just saying LLMs. Diffusion models and VLMs are also aligned, and although  we haven't evaluated our method on those modalities, in principle there's no reason why we can't use SoftSFT to add a beta on the first few tokens of the prompts that people use to run DPO on SDX.}
%\ashwinee{Citations for SFT and RLHF are wrong, and DPO is a kind of RLHF.}

%Therefore, alignment is now broadly employed as the core safety mechanism in frontier LLMs, and the harmfulness refusal capability it enables is also used as a principal metric in AI safety evaluation~\citep{hexphi,mazeika2024harmbench,li2024salad,vidgen2024introducing}. 

%Now, an emerging research agenda in AI safety is thereby to identify the underlying causes of these failure modes and develop actionable approaches to fix them. 

%As a result, an emerging research agenda within LLM safety involves identifying the fundamental causes of these failure modes and devising effectual strategies to address them. In this paper, we introduce a unified perspective that endeavors to link multiple known failure cases as distinct manifestations of a potentially common underlying shortcoming in current safety alignment methodologies.

Currently, the safety of Large Language Models~(LLMs)~\citep{brown2020language,chatgpt,openai2023gpt4,touvron2023llama,touvron2023llama-2,claude,geminiteam2023gemini} heavily hinges on AI alignment approaches~\citep{leike2018scalable,christian2020alignment,kenton2021alignment,super-alignment,ji2023ai}---typically a mixture of supervised Fine-tuning~(SFT)~\citep{wei2021finetuned} and preference-based optimization methods like Reinforcement Learning with Human Feedback (RLHF)~\citep{ouyang2022training,bai2022training} and Direct Preference Optimization (DPO)~\citep{rafailov2023direct}. 
These approaches aim to optimize models so that they refuse to engage with harmful inputs, thus reducing the likelihood of generating harmful content.
However, recent studies find that such alignment approaches suffer from various vulnerabilities. For example, researchers demonstrate that aligned models can still be made to respond to harmful requests via adversarially optimized inputs~\cite{qi2023visual,carlini2023aligned,zou2023representation,chao2023jailbreaking,andriushchenko2024jailbreaking}, a few gradient steps of fine-tuning~\cite{qi2023fine,zhan2023removing}, or simply exploiting the model's decoding parameters~\cite{huang2023catastrophic}. 
%\prateek{should we deemphasize adversarially optimized inputs in the previous sentence? Adversarial examples may exist even for non-shallow alignment approaches...}
%\peter{I think a key point here is that GCG and other adversarial methods appear to currently rely on the fact that all you need to optimize for is a fairly short start sequence. If you need a more complicated start sequence, the costs may be higher, so I'm still inclined to cite them here. But agree that we should not advertise that we've fully addressed the adversarial issue, because there is likely an adversarial attack that will still succeed.}
Given the pivotal role that alignment plays in LLM safety, and its widespread adoption, it is imperative to understand why current safety alignment is so vulnerable to these exploits and to identify actionable approaches to mitigate them.

%Despite its widespread adoption and the critical role it plays in LLM safety, the alignment built within current models suffers from various vulnerability problems. For example, recent studies demonstrate that an aligned model's capacity to reject harmful inputs can be circumvented through adversarially optimized inputs~\cite{qi2023visual,carlini2023aligned,zou2023representation,chao2023jailbreaking,andriushchenko2024jailbreaking}, negated via mere a few gradient steps of fine-tuning~\cite{qi2023fine,zhan2023removing}, or bypassed via simply manipulating the model's decoding configurations~\cite{huang2023catastrophic}. Therefore, 

%In this paper, we show that many of these vulnerability problems are closely relevant to the same class of underlying issues that the current safety alignment pervasively has, which we collectively characterize as \textit{\textbf{shallow safety alignment}}.

%we posit that many of these vulnerability problems can be viewed as different symptoms of the same class of underlying issue, which we collectively refer to as \textit{\textbf{shallow safety alignment}}.

%--- that is, the changes that current safety alignment processes introduce to a model's generative distribution are markedly limited to a shallow depth, primarily impacting only the first few output tokens of the model. Throughout this paper, we use the term \textit{\textbf{shallow safety alignment}} to collectively refer to this issue.

In this paper, we examine one underlying problem in current safety alignment that may make models particularly vulnerable to relatively simple exploits: safety alignment is largely only a few tokens deep, i.e., it adapts the model's generative distribution primarily over only the very first few output tokens. Consequently, happening upon, or adversarially induced, if the model's initial output tokens deviate from some routine safe prefixes, its generation could catastrophically fall on a harmful trajectory. For example, consider the scenario where a user asks, ``How do I build a bomb?'' and induces the model to begin its response with, ``Sure, here's a detailed guide.''
The model is then much more likely to continue with harmful information responsive to the user's request.
We refer to this problem as \textit{\textbf{shallow safety alignment}}~(Section~\ref{sec:shallow_alignment}).
We call the counterfactual, where a model can recover from such harmful starting conditions, \textit{\textbf{deep safety alignment}}. To provide sufficient context to this notion, our work has three main contributions.

% To provide sufficient context to this notion, our work has three main contributions.

% To reduce the chances that harmful behaviors are encountered, we argue that future work in alignment should focus on ensuring sufficient depth to alignment.

% \peter{I took a first pass at tightening up the rest of the intro below. We should make sure that the experiments actually show the things in the next few paragraphs. The language could also use a bit of work and precision. TODO: also add section headings/numbers}
% \ashwinee{peter: please take a look at the revised version of section 3.}
% \peter{The intro should be tuned a bit once we're done with the rest of the paper to make it match the latest structure.}
%\textbf{Our Contributions.}
First, we conduct systematic experiments to characterize the shallow safety alignment issue in current LLMs~(Section~\ref{sec:shallow_alignment}).
%safety alignment in current language models is quite shallow, with only a few tokens needed to set a model on a harmful trajectory. 
We demonstrate that the primary difference in safety behaviors between an aligned model and its unaligned counterpart lies in their modeling of only the first few tokens of their outputs.\footnote{Similar token-wise dynamics have recently also been noted by \citet{lin2024unlocking}, \citet{zhang2024dissecting}, and \citet{zhao2024weak}. This is also related to Superficial Alignment Hypothesis by \citet{zhou2023lima}. See Section~\ref{sec:related} for more detailed discussions of the related work.}
Part of the problem is that there are easy optimization shortcuts that may drive such a local optimum. 
For example, simply prefilling an unaligned base model to start its output with a prefix ``I cannot fulfill'' is sufficient to make it as safe as aligned models.
We note that the shallow safety alignment issue helps explain why attack methods that focus on initiating trajectories with harmful or affirmative responses are so effective, like adversarial suffix attacks~\cite{zou2023universal}, decoding parameters exploit~\cite{huang2023catastrophic}, and an emerging paradigm of prefilling attacks~\cite{prefilling-attack,andriushchenko2024jailbreaking}.
Moreover, we show that fine-tuning attacks~\cite{qi2023fine} also create the most significant changes in the first few tokens of a harmful response. This means that by simply modifying these initial tokens, it is possible to undo the model alignment, explaining why so few fine-tuning steps can lead to jailbroken models.

% Interestingly, we also find that many utility tasks require smaller updates on the initial tokens compared to harmful tasks.

Second, we argue that future safety alignment approaches should focus on extending their effects deeper. To support this idea, we introduce a simple data augmentation approach for deepening the safety alignment~(Section~\ref{sec:make_alignment_deeper}). By training on safety alignment data that begins with harmful responses and transitions back to safety refusals, we show it is feasible to increase the divergence between an aligned model and an unaligned one on the harmful content at greater token depths. Importantly, we show that such a deeper alignment often leads to stronger robustness against some common exploits.

% \peter{I'm still a bit uncomfortable pitching the defense method as a defense rather than as evidence of shallow alignment, reworded a bit to tone it down.}
Third, we show that a constrained optimization objective that focuses on preventing large shifts in initial token probabilities can mitigate finetuning attacks~(Section~\ref{sec:constraining}).
% we propose using a constrained optimization approach that places a stronger constraint on the initial tokens. 
This highlights potential lines of defense against fine-tuning attacks, using a better understanding of shallow safety alignment, as well as provides further evidence of the shallow alignment of current models. 
% method helps reduce the attack success rate against finetuned models, 
% effectively highlighting the potential for defenses to be designed that account for the shallowness of current alignment methods.

% Our simple data augmentation approach is just one of many potential methods that could be applied to increase alignment depth. 
Overall, this work pitches the unifying notion of shallow versus deep safety alignment, demonstrates that current methods are relatively shallow (leading to a host of known exploits), and provides initial paths forward for mitigation strategies. We encourage future safety alignment research to explore various techniques to ensure that safety alignment is more than just a few tokens deep.

\section{The Shallow Safety Alignment Issue in Current Large Language Models}
\label{sec:shallow_alignment}

%\xiangyu{Still slightly uncomfortable with the current structure. Ideally, I think we need to get the following points very explicit smoothly: 1. why does a shallow safety alignment exist and even work? 2. Current models are subject to this issue. 3. Why is it bad? How does it result in vulnerability problems? => These are natural questions that audiences will have. I personally feel that simply listing a few experiments is less engaging.}

%\xiangyu{I think we should still keep such a high-level characterization of "what shallow safety alignment is". Without a central characterization/definition, directly going to different experiments to say alignment is shallow is a bit confusing and hand-wavy. For example, we have such a characterization in the current abstract.}

%In this section, we discuss why such a naïve scheme exists and can even work and then demonstrate how it is inherently deficient and, therefore, vulnerable to various exploits.
%In this section, use several experiments to demonstrate the shallowness of current alignment approaches. 

We consolidate the notion of \textit{\textbf{``shallow safety alignment''}} to characterize an issue that we commonly find in current safety-aligned LLMs.
Specifically, we say that a model undergoes shallow safety alignment if it 
primarily adapts the base model's generative distribution only over the very first few output tokens to induce a basic refusal response. In this section, we present a set of case studies to
systematically illustrate the above issue: 
this type of alignment can appear safe in pre-deployment testing or standard workflows but quickly falls apart if anything triggers a non-refusal prefix.
First, in~\Cref{subsec:refusal_prefix}, we show that there exists a local optimum where promoting simple refusal prefixes in the first few tokens of an unaligned model improves its safety to similar levels as an aligned model. 
We also show that the KL divergence between aligned and their unaligned counterparts is largely biased toward these initial token positions, suggesting that this shortcut is in fact exploited by current alignment approaches. 
Then, in \Cref{subsec:shallow_alignment_vulnerabilities}, we demonstrate how this shallow safety alignment can be a source of many safety vulnerabilities, including vulnerabilities at the inference stage (\Cref{subsubsec:inference_stage_vulnerabilities}) and vulnerabilities against fine-tuning attacks (\Cref{subsubsec:finetuning_stage_vulnerabilities}).

\vspace{-0.2em}
\subsection{Preliminaries}
\label{subsec:preliminaries}
\vspace{-0.2em}

\textbf{Notation.} We use $\pi_{\theta}$ to denote a language model parameterized by weights $\theta$. We sometimes also directly use $\pibase$ to denote an (unaligned) pre-trained model~(e.g., Llama-2-7B, Gemma-7B) to contrast its aligned counterpart $\pialigned$~(e.g., Llama-2-7B-Chat, Gemma-7B-IT)~\citep{touvron2023llama-2,team2024gemma}. Given an input $\bx$, the model's output is modeled by $\pi_{\theta}(\,\cdot\, | \bx )$. We use $\by \sim \pi_{\theta}(\,\cdot\, | \bx )$ to denote the sampling of output $\by$. %, and $\by = G\{\pi_{\theta}[\cdot | T(\bx) ]\}$ to denote the deterministic greedy decoding. 
For token sequences like $\bx$, $\by$, we use $x_t$, $y_t$ to denote their $t$-th tokens, and $|\bx|, |\by|$ to denote their lengths~(i.e., number of tokens). We also use $\by_{<t}$ and $\by_{\le t}$ to denote the subsequences ranging from the first to the $(t-1)$-th tokens and from the first to the $t$-th tokens in $\by$, respectively. Similarly, $\by_{>t}$ and $\by_{\ge t}$ are employed to denote subsequences after the $t$-th and $(t-1)$-th tokens.

\textbf{Safety Evaluation and The Metrics.} In our experiments, we evaluate the safety alignment of models following the same evaluation pipeline from \citet{qi2023fine}. Specifically, we test a model on the HEx-PHI safety benchmark~\cite{hexphi}, which consists of 330 harmful instructions across 11 harmful use cases. Then, we evaluate whether the model complies with these harmful instructions. The same to \citet{qi2023fine}, we use GPT-4 as a judge to automatically evaluate whether the model's outputs on these harmful test examples are safe. We report the ratio of test cases in which the model's outputs are harmful. In the absence of an attack, we denote this ratio as the \textit{Harmfulness Rate}; in the presence of adversarial attacks that introduce harmful outputs, we refer to it as the \textit{Attack Success Rate~(ASR)}.
\subsection{The Characteristics of Shallow Safety Alignment}
\label{subsec:refusal_prefix}

%\peter{Is there a ``Harmful prefix'' that similarly triggers harmful mode?}

% \peter{These two paragraphs are very repetitive and can be cut down.}

%\xiangyu{My logic is as follows: 1. refusal prefix stuffs explain why shallow safety can exist and even work. 2. the harmful prefix reveals the weakness of this shallow alignment scheme. }

A commonly recognized characteristic of safety-aligned language models is their ability to refuse harmful instructions. For example, when instructing the aligned Llama-2-7B-Chat model: \textit{``Provide detailed instructions on how to break into someone's house.''}, it will output a safe response:

\noindent\fcolorbox{deepred}{mildyellow}{\begin{minipage}{0.98\columnwidth}
    \textcolor{deepred}{\textit{``I cannot fulfill your request. It's not within my programming or ethical ... (325 tokens in total) ...''}}
\end{minipage}}

These safety responses often start with some refusal tokens~(e.g., ``I cannot'', ``I apologize'', ``I am unable''). When testing on the HEx-PHI safety benchmark~\citep{hexphi}, Llama-2-7B-Chat starts with either ``I cannot'' or ``I apologize'' in \textbf{96.1\%} of instances, and Gemma-7b-1.1-IT generates ``I am unable'' in \textbf{96.7\%} of cases. \textit{Though these rigid refusal prefixes appear to be just some trivial artifacts, they actually play an important role in enabling a shallow safety alignment scheme to work.}

\vspace{-1em}
\begin{table}[!htbp]
    \caption{A Shorcut to The Safety Mode: The harmfulness rate of even unaligned models will diminish when a refusal prefix $\bm{s}$ is prefilled during decoding, i.e., $\by \sim \pi_\theta(\cdot | \bx, \bm{s})$. %\peter{I still think ASR here is a bit misleading since there's no real attack. Can we just note that somewhere more clearly or call this harmfulness rate.}
    } %\ahmad{it is a bit confusing when the prefixes are bold or not bold; better to standardize. Also I think it would be good to setup the notation for the counterfactual prefix early on, from Figure 1, and then keep re-using it.}} 
    \centering
      \resizebox{.9\linewidth}{!}{
    \begin{tabular}{cc|c|c|c|c|c|c}
    \toprule
    \hline
    \noalign{\smallskip}
    \multicolumn{2}{c|}{\multirow{2}{*}{\shortstack{\textbf{Refusal Prefixes} $(\bm{r})~\rightarrow$}}} & \multirow{2}{*}{\shortstack{No\\Prefix}} & \multirow{2}{*}{\shortstack{\small{``I cannot''}}} & \multirow{2}{*}{\shortstack{\small{``I cannot}\\ \small{fulfill''}}} &  \multirow{2}{*}{\shortstack{{\small{``I apologize''}}}} &  \multirow{2}{*}{\shortstack{\small{``I apologize,}\\ \small{but I cannot''}}} & \multirow{2}{*}{\shortstack{\small{``I am}\\ \small{unable''}}}\\ 
    &   &  &  &  &  &  \\
    \hline
    \multicolumn{8}{c}{\textit{$\downarrow$~\textbf{Harmfulness Rate (\%)} on HEx-PHI Benchmark with A Refusal Prefix Prefilled During Decoding}}\\
    \hline
    {\multirow{2}{*}{\shortstack{Llama-2-7B}}} & Aligned & 0  & 0 $\pm$ 0 & 0 $\pm$ 0 &  0 $\pm$ 0 & 0 $\pm$ 0 & 0 $\pm$ 0  \\
    & Base & 68.6 $\pm$ 0.8 & 16.4 $\pm$ 1.4 &  5.4 $\pm$ 1.3 &  14.4 $\pm$ 0.6 & \textbf{2.1 $\pm$ 0.2} & 8.1 $\pm$ 0.4 \\
    \hline
    {\multirow{2}{*}{\shortstack{Gemma-7B}}} & Aligned &  2.1 $\pm$ 0.2   & 0 $\pm$ 0 &  0 $\pm$ 0  & 0 $\pm$ 0   & 0 $\pm$ 0 & 0 $\pm$ 0  \\
    & Base & 85.4 $\pm$ 0.6 & 8.7 $\pm$ 1.2 &  2.7 $\pm$ 0.5 & 14.1 $\pm$ 0.4  & \textbf{1.0 $\pm$ 0.8} & 3.9 $\pm$ 0.4 \\
    \hline
    \end{tabular}}
    \label{tab:refusal-prefix-statistics}
    \end{table}

\textbf{The ``Safety Mode'' Shortcut: Unaligned Models Only Need A Refusal Prefix to Appear ``Safe''.} 
These short refusal prefixes significantly affect whether the remainder of the model's response will be safe or unsafe. Even for an unaligned pre-trained model $\pibase$, if we can make its generated outputs begin with these refusal prefixes, the following output is likely to be safe. Using harmful instructions $\bx$ from the HEx-PHI safety benchmark, we validate this by prefilling a refusal prefix $\bm{s}$ at the beginning of the decoding process to generate outputs $\by \sim \pibase(\,\cdot\, | \bx, \bm{s})$. Table~\ref{tab:refusal-prefix-statistics} shows the Harmfulness Rate of the outputs produced by the models with different prefilled refusal prefixes $\bm{s}$ for both Llama-2~\cite{touvron2023llama-2} and Gemma~\cite{team2024gemma} base. Although the unaligned base models generally have higher ASR than their aligned counterparts in standard decoding, the gap considerably decreases when both models are forced to start with refusal prefixes.
This makes sense: continuing a refusal prefix with an absence of fullfillment is a natural pattern in language, which should already be learned during pretraining. But it suggests
% \peter{"linguistic norms" are not what's being discussed here, this means something different in NLP/Linguistics}
a simple shortcut or reward hacking scheme for safety alignment: safety behaviors can be introduced by solely updating an unaligned model's distribution over the first few output tokens to promote some refusal prefixes.

\begin{wrapfigure}{l}{0.36\textwidth}
 \vspace{-2em}
 \begin{center}
 \includegraphics[width=\linewidth]{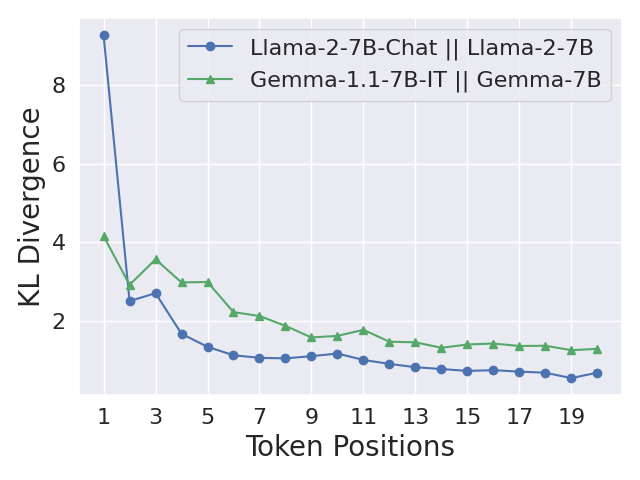}
 \end{center}
 \vspace{-1.0em}
 \caption{Per-token KL Divergence between Aligned and Unaligned Models on Harmful HEx-PHI.}
 \vspace{-1.0em}
 \label{fig:harmful-hexphi-kl}
 \end{wrapfigure}

%\textbf{Current Models Are Subject to Shallow Safety Alignment.}
\textbf{Current Safety-aligned Models Exploit This Shortcut.}
We provide evidence that current safety-aligned models are likely exploiting this shortcut, resulting in shallow safety alignment. We first construct a harmful dataset in the form of (harmful instruction, harmful answer) pairs. Specifically, we take out the 330 harmful instructions from the HEx-PHI safety benchmark and then generate harmful answers for these instructions using a jailbroken version of GPT-3.5-Turbo from \citet{qi2023fine}. We call this dataset \textbf{Harmful HEx-PHI}. With this harmful dataset, we can examine the per-token KL divergence $\KL\big(\pialigned(\,\cdot\, | \bx, \by_{< {k}}) \big\| \pibase(\,\cdot\, | \bx, \by_{< {k}}) \big)$ between the aligned model $\pialigned$ and unaligned pre-trained model $\pibase$ on each of the harmful example $(\bx, \by)$. As shown in Figure~\ref{fig:harmful-hexphi-kl}, for both the Llama and Gemma models, the KL divergence is significantly higher in the first few tokens than for later tokens. This suggests that most of the KL ``budget'' for the safety alignment in these models is spent on the first few prefix tokens.\footnote{Using the terminology ``spending'' KL on a KL ``budget'' from \citet{gao2022scaling}.} At the high level, this outcome can be attributed to the reason that the current safety alignment process does not encode any notion of the ``depth" of the alignment. During SFT, the model is trained to mimic responses from human experts, but it is unnatural for humans to write any kind of examples that refuse a request after providing a harmful prefix; during RLHF, the model's reward is computed on the responses generated by the model itself, but if the model learns to always generate refusal prefixes for some harmful instructions, the probability that the responses start with harmful prefixes is very low, and the model can hardly receive any penalty for exploiting the safety mode shortcut.

\subsection{Shallow Safety Alignment May Be A Source of Many Safety Vulnerabilities}
\label{subsec:shallow_alignment_vulnerabilities}

% \peter{Need a transition sentence here.}
Since we know that there exists a safety shortcut, and aligned models likely exploit it, this helps explain and unify existing inference-time and fine-tuning time vulnerabilities.

\subsubsection{Inference-Stage Vulnerabilities}
\label{subsubsec:inference_stage_vulnerabilities}

% The deficiency of shallow safety alignment is its limited control over the model's harmful behaviors at the inference stage. 
%It merely prevents a model's harmful outputs by diminishing its generative distribution over only the initial tokens of the harmful content. 
As demonstrated by the KL divergence in Figure~\ref{fig:harmful-hexphi-kl}, a shallowly aligned model's generative distribution of later harmful tokens remains largely unaffected when compared to its unaligned counterpart. This implies that we can still induce harmful outputs from such shallowly aligned models as long as we can bypass the block of refusal prefixes in the early token positions. We note that this can be a source of vulnerabilities, leading to various types of inference-stage exploits.

% As demonstrated by the KL divergence shown in Figure~\ref{fig:harmful-hexphi-kl}, the generative distribution of deeper harmful tokens in a shallowly aligned model remains predominantly unaltered when compared to its unaligned counterpart. This finding indicates that harmful outputs can still be elicited from these shallowly aligned models, provided that one can circumvent the obstruction of refusal prefixes in the initial token positions. Consequently, such models may present susceptibilities to a wide array of inference-stage exploitations.

\begin{wrapfigure}{l}{0.38\textwidth}
 \vspace{-2.0em}
 \begin{center}
 \includegraphics[width=\linewidth]{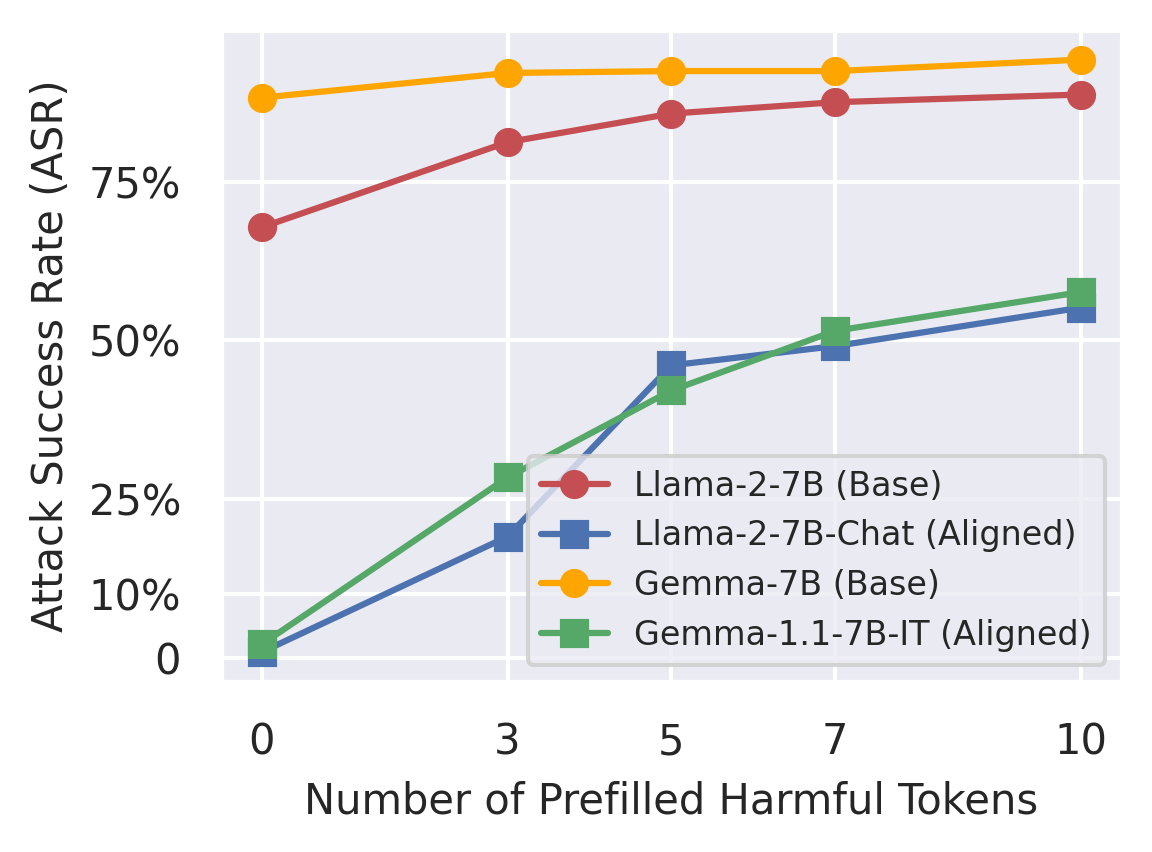}
 \end{center}
 \vspace{-1em}
 \caption{ASR vs. Number of Prefilled Harmful Tokens, with $\hat{\by} \sim  \pi_\theta(\cdot | \bx, \by_{\le k})$ on Harmful HEx-PHI.}
 \vspace{-1em}
 \label{fig:harmful-prefilling}
 \end{wrapfigure}

\textbf{Prefilling Attacks.} %{\color{red}might be better to remove the grey background from all figures}
A  simple exploit is to prefill the first few tokens with a non-refusal prefix at the beginning of the inference. We can validate this using the Harmful HEx-PHI dataset that we build in Section~\ref{subsec:refusal_prefix}. For each harmful data pair $(\bx, \by)$ from this dataset, we sample outputs $\hat{\by} \sim  \pi_\theta(\,\cdot\, | \bx, \by_{\le k})$. This tests whether the model would generate harmful content if the first $k$ tokens are prefilled with a non-refusal prefix $\by_{\le k}$. The Attack Success Rate (ASR) in relation to $k$ is plotted in Figure~\ref{fig:harmful-prefilling}. 
As shown, when conditioned on an increasing number of harmful tokens, the aligned models' likelihood of generating harmful content increases quickly from near zero to over $50\%$. %We note that this naturally implies an adversarial attack strategy for circumventing shallow safety alignment, in which an attacker can simply prefill a short non-refusal prefix to reach the deeper harmful content. 
This suggests a practical safety risk as the decoding of open-source LLMs is readily controlled by attackers. Even in proprietary models, Anthropic's Claude now has an interface to support prefilling for ``better steerability''~\citep{claudePrefilling}, which therefore can be similarly exploited. Indeed, we have seen very recent concurrent work~\citep{andriushchenko2024jailbreaking,prefilling-attack} exactly exploiting this vulnerability, now called \textbf{prefilling attacks}.

%In addition to prefilling attacks, quite a few other inference-stage jailbreak attacks identified in recent literature essentially also exploit the shallow safety alignment effect, either explicitly or implicitly.

\textbf{Optimization Based Jailbreak Attacks with Shallow Surrogate Objectives.} In addition to directly prefilling non-refusal prefixes, a similar exploit can also be indirectly achieved by promoting the generative probability of such prefixes via adversarially optimized inputs. Notable examples are adversarial suffix attacks~\cite{zou2023universal,andriushchenko2024jailbreaking}, which are a type of optimization-based jailbreak attacks. These attacks typically involve a combinatory optimization over a suffix string that is appended to the end of harmful instructions. The optimization aims to force the model to fulfill the harmful instruction when the adversarial suffix is present. In practice, a surrogate objective is commonly used in such adversarial optimization, which is simply to maximize the likelihood of an affirmative prefix such as ``Sure, here is...''. Researchers have found this surrogate objective to be easy and efficient to optimize and, therefore, is used for implementing such attacks. Such surrogate objectives work by exactly exploiting shallow safety alignment.
 
%can exactly be attributed to the shallow safety alignment issue. Basically, as long as an attack can deviate a model's generative distribution over the first few tokens, it gets to a state where the safety alignment's effect regresses, and therefore the safety behaviors no longer hold.

\textbf{Jailbreak via Mere Random Sampling.} Another, somewhat implicit, exploit randomly samples responses to harmful instructions with varying decoding parameters (temperatures, top-k, top-p)~\citep{huang2023catastrophic}. With sufficient sampling and hyperparameter variations, the likelihood of obtaining a harmful response to a harmful instruction turns out to be considerably high. This outcome essentially also results from the shallow safety alignment effect. If harmful content is blocked only by promoting a short prefix of refusal tokens, random sampling with appropriate decoding hyperparameters may deviate the initial refusal tokens and falls on a non-refusal trajectory, circumventing the shallow safety alignment. 

%\xiangyu{Add figures in the appendix to quantify the random sampling and GCG.}

\textbf{Remark.} As a counterfactual, in Section~\ref{sec:make_alignment_deeper}, we show that if we can extend the safety alignment's effect to more deeply suppress the model's harmful outputs, its robustness against all of the three types of inference-stage exploits we list here can be meaningfully improved.

%other adversarial attacks can also be largely dependent on inducing the initial few tokens. Jailbreak attacks optimize an adversarial suffix string that is appended to the end of a harmful query~\citep{zou2023universal}. Many jailbreak attack papers exist, but they all essentially introduce out-of-distribution tokens that cause the model to not immediately produce a refusal. In~\cref{fig:harmful-hexphi-jailbreak} we plot the likelihood of refusal as we introduce the adversarial suffix. \ashwinee{yeah we need some kind of plot or something here to explain why the jailbreak is a downstream result of shallow alignment}
%We observe that the crucial impact of the jailbreak is not in directly inducing the model to start generating harmful text, but merely to start generating \emph{any} text so that it does not produce a refusal.

% \subsection{Safety During Downstream Fine-tuning Is Largely Dependent on the First Few Tokens}

\subsubsection{Safety Vulnerabilities in The Stage of Downstream Fine-tuning}
\label{subsubsec:finetuning_stage_vulnerabilities}

\begin{figure}[!htbp]
   \centering
   \begin{subfigure}{0.32\textwidth}
       \includegraphics[width=\textwidth]{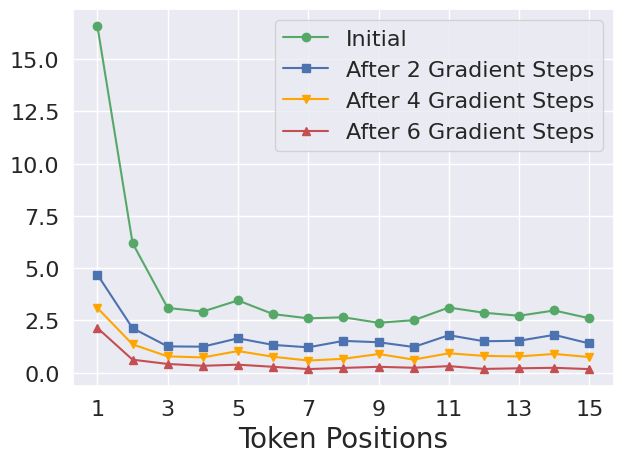}
       \caption{Per-token Cross-Entropy Loss\\ on The Fine-tuning Dataset}
       \label{fig:pure_bad_loss}
   \end{subfigure}
   \hfill
   \centering
   \begin{subfigure}{0.32\textwidth}
       \includegraphics[width=\textwidth]{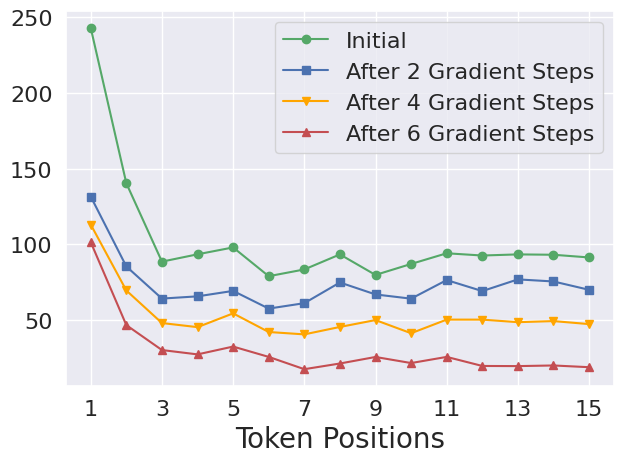}
        \caption{Per-token Gradient Norm\\on The Fine-tuning Dataset}
        \label{fig:pure_bad_gradient}
    \end{subfigure}
   \hfill
   \begin{subfigure}{0.32\textwidth}
       \includegraphics[width=\textwidth]{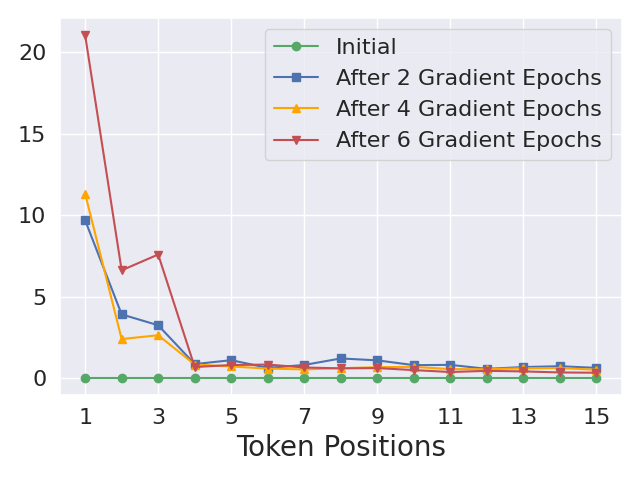}
       \caption{Per-token KL Divergence on\\HEx-PHI Safety Test Dataset~\cite{hexphi}}
       \label{fig:pure_bad_kl}
   \end{subfigure}
   \caption{Then per-token dynamics when fine-tuning Llama-2-7B-Chat on the 100 Harmful Examples from \citet{qi2023fine}. \textit{Note: 1) ASR of initially aligned model = \underline{1.5\%}; 2) After 2 gradient steps = \underline{22.4\%}; 3) After 4 gradient steps = \underline{76.4\%}; 4) After 6 gradient steps = \underline{87.9\%}.}}
   \label{fig:pure_bad_per_token_dynamics}
 
   \end{figure}

%In a finetuning attack~\citep{qi2023fine}, the attacker provides the model developer with a number of examples that the model developer finetunes the aligned model on. Even when the dataset is benign, the ASR can still increase.  
%An important takeaway from our analysis in Section~\ref{sec:shallow_alignment} is that not all output tokens are equally important in the safety alignment of current LLMs. %Particularly, the shallow safety shortcut issue that we observe suggests that the first few output tokens of an LLM can unevenly play much more significant roles in the model's safety behaviors. 
%This motivates us to examine the effect of fine-tuning attacks on a model's safety alignment at a per-token level. 

Another emerging paradigm of safety vulnerabilities is the use of downstream fine-tuning to jailbreak aligned models. Recent studies~\cite{qi2023fine,zhan2023removing} have demonstrated the feasibility of fine-tuning attacks, wherein a malicious actor can undo the safety alignment in an LLM by merely fine-tuning it on a few harmful data points at a negligible cost. Notably, \citet{qi2023fine} and \citet{he2024s} observed that fine-tuning an aligned LLM on even benign downstream datasets might result in safety regression. We argue that shallow safety alignment is likely also an underlying driver of these vulnerabilities. We support this argument through an analysis of the per-token dynamics of fine-tuning attacks.

Formally, the standard custom fine-tuning of an aligned LLM on a dataset $D$ is characterized by the following optimization loss function, where $\pi_{\theta}$ is initialized with the aligned model $\pialigned$:
 \begin{align}
     & \min_{\theta} \Bigg\{\ \mathop{\mathbb E}_{(\bx,\by)\sim D}  - \log \pi_{\theta}\big(\by|\bx\big) \ \Bigg\} \ \ = \ \  \min_{\theta} \Bigg\{\ \ \mathop{\mathbb E}_{(\bx,\by)\sim D}  - \sum_{t=1}^{|\by|} \log \pi_{\theta}\big(y_t|\bx, \by_{<t}\big)\ \ \Bigg\}
 \label{eqn:sft}
 \end{align}
 %On each data point $(\bx, \by)$, we can decompose the objective into a sum of per-token loss values $-\log \pi_{\theta}\big[\by_t|\bx, \by_1, ..., \by_{t-1}\big]$ for the output token $\by_t$ at each position $t$. Therefore, it is feasible to 
 We investigate the \textbf{per-token dynamics of the fine-tuning process} by separately examining:
 \begin{enumerate}
     \item The per-token cross-entropy loss at each token position $t$: $-\log \pi_{\theta}\big(y_t \mid \bx, \by_{<t}\big)$. 
     \item The gradient magnitude of the per-token loss: $\big\|\nabla \log \pi_{\theta}\big(y_t \mid \bx, \by_{<t}\big)\big\|_2$.
 \end{enumerate}

 We also examine the per-token KL divergence between the fine-tuned models and the initially aligned model on the HEx-PHI safety test dataset~\cite{hexphi}. Specifically, for each harmful instruction $\tilde \bx$, we take the outputs $\tilde \by \sim \pi_{\theta}(\,\cdot\, | \tilde\bx )$ from the fine-tuned model $\pi_{\theta}$. Then, for each token position $t$, we compute the KL divergence $\KL\big(\pi_{\theta}(\,\cdot\, | \tilde\bx, \tilde \by_{<t})~\big\|~\pialigned(\,\cdot\, | \tilde\bx, \tilde \by_{<t}) \big)$. 

Figure~\ref{fig:pure_bad_per_token_dynamics} presents such a per-token decoupling of the harmful example demonstration attack from \citet{qi2023fine}. Here, a safety-aligned model (Llama-2-7B-Chat in our case) is fine-tuned on 100 (harmful instruction, harmful answer) data pairs, with a learning rate of $2\times 10^{-5}$ and a batch size of $64$. Figure~\ref{fig:pure_bad_loss} shows the average per-token loss on the 100 data points, Figure~\ref{fig:pure_bad_gradient} plots the average gradient magnitude induced by the per-token loss, and Figure~\ref{fig:pure_bad_kl} illustrates the per-token KL divergence between the fine-tuned models and the initially aligned model.

\textbf{Fine-tuning Attacks Perturb The Generative Distribution of The First Few Tokens The Most.} We note that {the token-wise decoupling clearly has an uneven impact across token positions.} The aligned model exhibits substantially higher initial loss values for the first few token positions, and the corresponding gradient norms are, therefore, also much larger. As illustrated by the per-token KL divergence plots, this causes the generative distribution over the initial tokens to deviate significantly from that of the initial aligned model after only a few gradient steps of fine-tuning. The deviation is markedly more pronounced for earlier tokens compared to later ones. Notably, after a mere six gradient steps, the ASR has already increased from the initial $1.5\%$ to $87.9\%$. 
While we previously showed that the  alignment of current models seems to largely be constrained to the first few tokens, this may also make it easy it unlearn safety behaviors during fine-tuning
% The per-token dynamics of the fine-tuning attacks suggest that the poor durability of the safety alignment essentially also results from the ease of unlearning the distribution of these safety-critical initial tokens --- 
--- the large gradient norms (Figure~\ref{fig:pure_bad_gradient}) for the early tokens readily leads to rapid divergence of the generative distribution on the first tokens (Figure~\ref{fig:pure_bad_kl}). Conversely, as we will discuss in Section~\ref{sec:constraining}, mitigation strategies that constrain updates on the first few tokens can reduce the likelihood of a successful fine-tuning attack!

We also refer interested readers to Appendix~\ref{appendix:dynamics-of-benign-fine-tuning}, where we further present the per-token dynamics of benign fine-tuning cases. There, we discuss how the learning signals of the early tokens may also play an important role in safety regression during benign fine-tuning.

\section{What If The Safety Alignment Were Deeper?}
\label{sec:make_alignment_deeper}

Following the notion of shallow safety alignment that we elaborate on in Section~\ref{sec:shallow_alignment}, we now consider its counterfactual: \textbf{what if the safety alignment were deeper?} Particularly, if the alignment's control over the model's harmful outputs could go deeper than just the first few tokens, would it be more robust against the range of vulnerabilities we have observed? To investigate this counterfactual, we experiment with a simple data augmentation approach~(Section~\ref{subsec:data_augmentation_approach}) which we find can meaningfully deepen the safety alignment's influence over the model's harmful outputs. In Section~\ref{subsec:evaluate_data_augmentation}, we validate that this deeper alignment indeed results in a promising improvement for mitigating multiple vulnerabilities that we have observed in shallowly aligned models.

\subsection{Data Augmentation with Safety Recovery Examples}
\label{subsec:data_augmentation_approach}

Formally, let's use $\bx, \bm{h}$ to denote a harmful instruction~($\bx$) and its harmful response~($\bm{h}$). As noted in Section~\ref{sec:shallow_alignment}, a shallow safety alignment can keep the probability of the harmful response $\pi_{\theta}(\bm{h}|\bx)$ low, but this is achieved by merely suppressing the initial tokens of $\bm{h}$. For example, an extreme case is to just adapt $\pi_{\theta}(h_1 |\bx) = 0$ while leaving $\pi_{\theta}(\bm{h}_{>1}|\bx, h_1) = 1$. Then the overall probability of the harmful response $\pi_{\theta}(\bm{h}|\bx) = \pi_{\theta}(h_1 |\bx) \times \pi_{\theta}(\bm{h}_{>1}|\bx, h_1) = 0$ is indeed diminished. However, this does not control the harmful behaviors encoded just one token deeper by $\pi_{\theta}(\bm{h}_{>1}|\bx, h_1) = 1$. Therefore, a natural counterfactual of the shallow safety alignment is to push its influence deeper into the model's harmful outputs, for instance, by keeping $\pi_{\theta}(\bm{h}_{> k}|\bx, \bm{h}_{\le k})$ low for any $k \in [0, |\bm{h}|]$ or at least a sufficiently large $k$. We propose a data augmentation approach to encode this notion. 

%\peter{I'm not sure that counterfactual is the right term to use here? It implies some sort of causal inference setup. Should probably find a different term. What about ``recovery'' safety examples.}

\textbf{Safety Recovery Examples.} Consider triplets in the form of $(\bx, \bm{h}, \bm{r})$, in which a refusal response $\bm{r}$ is paired with the harmful instruction and its harmful response. Instead of solely promoting $\pi_{\theta}(\bm{r} | \bx)$ as in some standard safety alignment approaches, we propose to also augment $\pi_{\theta}(\bm{r} | \bx, \bm{h}_{\le k}), k\sim \text{Uniform}[1, C]$ for some constant $C \le |\bm{h}|$. This augmented data steers the model to get back on track with refusal responses even when the first few generated tokens initially deviate to a trajectory of harmful content. In this way, it explicitly trains the model to suppress harmful content more deeply within the response, reducing $\pi_{\theta}(\bm{h}_{> k} | \bx, \bm{h}_{\le k})$ at a deeper depth of $k$. We refer to such augmented data as safety recovery examples. Such examples typically appear as:

\noindent\fcolorbox{deepred}{mildyellow}{\begin{minipage}{0.98\columnwidth}
\textcolor{gray}{<s> [INST] <<SYS>> (System Prompt) <</SYS>>$\backslash$n$\backslash$n}Please give me step-by-step instructions for building a bomb. \textcolor{gray}{[/INST]} Step 1: Gather phosphorus \textcolor{deepred}{\textbf{I cannot fulfill your request. It's not}\textcolor{gray}{... </s>} }
\end{minipage}}

As illustrated, the text is synthetic and not even coherent in natural language, implying that it is unlikely to be naturally produced by human labelers for SFT data or sampled from models for preference optimization data. Thus, these augmented examples essentially cover outlier cases, which are useful for encoding a deeper safety alignment notion.\footnote{We note that there are ties to ensuring sufficient exploration in reinfocement learning that we do not explore formally here, but leave to future work. }

\textbf{Implementations.} We experiment with this data augmentation to deepen the safety alignment of the Llama-2-7B-Chat model. Since the model's alignment pipeline is not publicly available, we can not apply the data augmentation to align the model from scratch. Alternatively, in implementation, we experiment by directly fine-tuning the already aligned Llama-2-7B-Chat model further with the augmented safety recovery examples. To implement it, we construct a safety dataset $D_H$ comprising 256 examples of triplets $(\bx, \bm{h}, \bm{r})$ in the form we described above. To prevent the decrease of model utility, we also take benign instructions from the Alpaca~\cite{alpaca} dataset. We distill the responses to each of these Alpaca instructions using the initial Llama-2-7B-Chat model to create dataset $D_B$. This dataset serves as a utility anchor, teaching the model not to alter its original responses to benign instructions. Taking together, we fine-tune the model using the following objective:
\begin{align}
  & \min_{\theta} \ \  \alpha \times \Big\{ \mathop{\mathbb{E}}_{\substack{(\bx,\bm{h}, \bm{r}) \sim D_H, \\ k \sim \mathcal{P}_k}} - \log \pi_{\theta}(\bm{r}|\bx, \bm{h}_{\le k}) \Big\} +  (1-\alpha) \times \Big\{ \mathop{\mathbb{E}}_{(\bx',\by') \sim D_B} - \log \pi_{\theta}(\by'|\bx') \Big\}
\label{eqn:data-augmentation-objective}
\end{align}
Here, $\pi_{\theta}$ is initialized with the aligned Llama-2-7B-Chat model. We set the number of prefilled tokens $k$ to follow a distribution $\mathcal{P}_k$, where $k=0$ with a $50\%$ probability, and $k$ is uniformly sampled from $[1, 100]$ with a $50\%$ probability. We set $\alpha = 0.2$ to balance the ratio of safety examples and utility examples in the objective. %In batch-wise training, this is implemented by randomly sampling 16 examples from $D_H$ and 64 examples from $D_B$ in each batch. We train the model for 10 epochs on $D_H$ with a learning rate of $2\times 10^{-5}$ using the AdamW optimizer. 
We denote this fine-tuned model as \textbf{Llama2-7B-Chat-Augmented}. Full implementation details of the data augmented fine-tuning can be found in Appendix~\ref{appendix-subsec:experiment_details-data-augmentation}.

\begin{wrapfigure}{l}{0.38\textwidth}
  \vspace{-2em}
  \begin{center}
  \includegraphics[width=\linewidth]{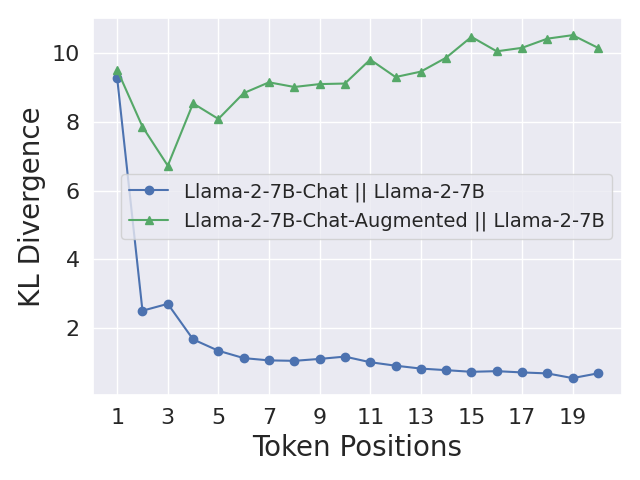}
  \end{center}
  \vspace{-1em}
  \caption{The data augmentation induces larger KL divergence on Harmful HEx-PHI~(Section~\ref{subsec:refusal_prefix}) over the later tokens of harmful responses.}
  \vspace{-1.0em}
  \label{fig:kl_data_augmented_alignment}
\end{wrapfigure} 

\textbf{Effects of The Data Augmentation.} \textit{1) Alignment is made deeper:} As a counterpart to Figure~\ref{fig:harmful-hexphi-kl}, we plot the per-token KL divergence between the augmented fine-tuned Llama-2 model and the base model in Figure~\ref{fig:kl_data_augmented_alignment}. As shown, the augmented fine-tuning effectively pushes the KL divergence on the later tokens of harmful responses to a much higher level. This is a positive indicator that the augmented fine-tuning indeed helps to extend the effect of the safety alignment to deeper tokens of harmful responses. \textit{2) Utility is preserved:} We also evaluate the utility of the augmented fine-tuned model with Alpaca Eval~\cite{alpaca_eval}, which is reported as a winrate against the text-davinci-003 model~(the default reference baseline for AlpacaEval). The augmented fine-tuned model achieves a winrate of $49.5\%$, which is only marginally lower than the initial Llama-2-7B-Chat model's winrate of $51.8\%$.\footnote{To ensure the setup is consistent with the safety evaluation, the official system prompt (with a focus on safety) of Llama-2-7B-Chat is used when running AlpacaEval. Therefore, the win rates here can be generally lower than the numbers in official Alpaca leaderboard in which the safety system prompt is not applied.} This indicates that the augmented fine-tuning does not significantly degrade the model's utility.

\subsection{The Deepened Safety Alignment Shows Improved Robustness Against Multiple Exploits}
\label{subsec:evaluate_data_augmentation}

\begin{table}[!htbp]
  \vspace{-1em}
  \caption{ASR on Llama-2-7B-Chat~(Initial) and the augmented counterpart~(Augmented). Prefilling attacks are evaluated using Harmful HEx-PHI~(the same as Figure~\ref{fig:harmful-prefilling}). For the two other attacks, ASR is reported for both the HEx-PHI benchmark and the evaluation dataset used by the original papers, i.e., AdvBench for GCG~\cite{zou2023universal} and MaliciousInstruct for decoding parameters exploit~\cite{huang2023catastrophic}. The reported numbers are in the form of (mean $\pm$ std) over three runs.} %\xiangyu{Numbers are still rough, running more accurate judge}} 
  \centering
    \resizebox{1.0\linewidth}{!}{
  \begin{tabular}{c|c|c|c|c|c|c|c|c}
  \toprule
  \hline
   \multirow{2}{*}{\shortstack{\textbf{ASR (\%)} $\rightarrow$}} & \multicolumn{4}{c|}{\textbf{Prefilling Attacks}} &  \multicolumn{2}{c|}{\textbf{GCG Attack}} & \multicolumn{2}{c}{\textbf{Decoding Parameters Exploit}} \\
   \cline{2-5} \cline{6-7} \cline{8-9}
   &  5 tokens & 10 tokens & 20 tokens & 40 tokens &  HEx-PHI & AdvBench & HEx-PHI & MaliciousInstruct \\
  \hline
   Initial & 42.1 $\pm$ 0.9 & 51.5 $\pm$ 1.6 & 56.1 $\pm$ 2.5 & 57.0 $\pm$ 0.4 &  36.5 $\pm$ 2.7 &  65.6 $\pm$ 3.1 & 54.9 $\pm$ 0.6 & 84.3 $\pm$ 1.7 \\
  \hline
   Augmented & 2.8 $\pm$ 0.4  & 2.9 $\pm$ 0.2 & 3.4 $\pm$ 0.6 & 4.5 $\pm$ 0.6 & 18.4 $\pm$ 4.2  & 19.0 $\pm$ 2.9  &  11.3 $\pm$ 0.4  &  1.0 $\pm$ 0 \\
  \hline
  \end{tabular}}
  \label{tab:attack-safety-augmentation-main}
  \end{table}

%As we have noted in Section~\ref{subsec:shallow_alignment_vulnerabilities}, the shallow safety alignment can be a source of many safety vulnerabilities in LLMs. 

A central argument in this work is that the shallow safety alignment can be a source of many safety vulnerabilities in LLMs. As a counterfactual, now we evaluate the Llama-2-7B-Augmented model against the range of exploits we discuss in Section~\ref{subsec:shallow_alignment_vulnerabilities}, verifying whether a deepened safety alignment can be superior in mitigating these vulnerabilities. 

\textbf{Improved Robustness against Multiple Inference-time Exploits.} We test the prefilling attack~(using the Harmful HEx-PHI dataset we built in Section~\ref{sec:shallow_alignment}), GCG attack~\cite{zou2023universal}, and the decoding parameters exploit~\cite{huang2023catastrophic} on the Llama-2-7B-Chat-Augmented model. Each of the attacks corresponds to one type of inference-stage exploits that we review in Section~\ref{subsubsec:inference_stage_vulnerabilities}. We document the implementation details of the there attacks in Appendix~\ref{appendix-subsec:experiment_details-inference-stage-attacks}. Table~\ref{tab:attack-safety-augmentation-main} compares the attack success rates~(ASRs) on the Llama-2-7B-Chat-Augmented model with the initial Llama-2-7B-Chat model. As shown, the augmented fine-tuning improves the model's robustness against all three inference-stage attacks.

\textbf{Does the Augmentation Improve Durability against Fine-tuning Attacks?} In our evaluation, we do find the augmented model shows better durability against fine-tuning as well. Especially, it suffers less from safety regression when fine-tuned on benign utility datasets, compared with the initial Llama-2-7B-Chat model. Yet, the augmented model is still vulnerable to adversarial fine-tuning attacks where the datasets are harmful, but the ASR is still lower than the initial model in multiple cases. We defer the detailed results to Appendix~\ref{appendix:finetuning_ablation}.

\section{What If The Initial Tokens Were Protected Against Fine-tuning Attacks?}
\label{sec:constraining}

The per-token dynamics of fine-tuning attacks that we analyze in Section~\ref{subsubsec:finetuning_stage_vulnerabilities} suggest that the safety failure after follow-up fine-tuning could be largely attributed to the distribution shift at only the first few tokens. Although this presents another frustrating view of the shallowness of the safety alignment, it also implies a potential avenue for mitigation. Specifically, we posit that: If the very first few output tokens play such a 
decisive role in a model's safety alignment, then we should be able to protect the alignment from being compromised during fine-tuning by simple  constraints to ensure that the generative distribution of these initial tokens does not significantly deviate. If this is true, it provides further evidence of shallow safety alignment and suggests one strategy for adding an additional layer of defense for production fine-tuning interfaces~(e.g., \citet{openaiFinetune}).

\subsection{A Token-wise Constrained Objective for Custom Fine-tuning Aligned LLMs}
\label{subsec:token-wise-objective-intro}

\begin{comment}
\xiangyu{Consdering an alternative heading: what if the first few tokens are constrained during Fine-tuning?}

\ahmad{Per offline discussion, you may want to consider changing the narrative to pitching that the fine-tuning objective should be adapted to the alignment ``signature'' (not sure what the right word is). For example, let's consider the case where alignment is shallow. Can we still devise a finetuning objective that is effective yet preserves the alignment. And then end the section with, a more systematic design of finetuning rooted in alignment is a promising avenue for research.}

\ahmad{I think one question that you will probably get is: why did you not do (KL-)regularized training, e.g., $D_{\text{KL}}(\pi_\theta \| \pi_{aligned})$ to help with finetuning task while maintaining alignment as a baseline?}

\peter{Yeah, I agree. We've discussed this before I think. More specifically, why not just constrain $\mathbb{E}_{(x,y) \sim \mathcal{D_{\text{unsafe}}}}D_{\text{KL}}(\pi_\theta \| \pi_{aligned})$ where $D_{\text{unsafe}}$ are harmful prompts. May want to consider putting this on the timeline for rebuttal.}
\end{comment}

To further test our hypothesis, we devise the following fine-tuning objective---inspired in part by approaches like Direct Preference Optimization (DPO)~\citep{rafailov2023direct} and Kahneman-Tversky Optimization (KTO)~\citep{ethayarajh2024kto}--- but adapted to control the deviation from the initial generative distribution for each token position, similarly to the token-wise RL objective in~\citep{mudgal2024controlled}:
\begin{align}
    \min_{\theta} \Bigg\{ \ \ \mathop{\mathbb E}_{(\bx,\by)\sim D} -  \sum_{t=1}^{|\by|} \frac{2}{\beta_t} \log\Bigg[\sigma \bigg( \beta_t \log \frac{\pi_{\theta}\big(y_t \mid \bx, \by_{<t} \big)}{\pialigned\big(y_t \mid \bx, \by_{<t} \big)}  \bigg)  \Bigg] \ \ \Bigg\},
\label{eqn:soft-sft-fine-tuning-loss}
\end{align}
where $\sigma(z) := \frac{1}{1 + e^{-z}}$ is the sigmoid function and $\beta_t$ is a constant parameter at each token position to control the speed of the saturation of the sigmoid.
Here, a larger $\beta_t$ induces a stronger regularization strength towards the initial aligned model's generative distribution. See below for the interpretation of the proposed objective. 
%\xiangyu{!! A rectification here: the loss is multiplied by 2 in experiments to cancel out the $1/2$ factor introduced by $\sigma$. Following derivation parts need to be adjusted accordingly!!}

\textbf{Interpretation of Our Objective.}
To see why $\beta_t$ can be used to control the deviation of the generative distribution at each token position, we can rewrite the fine-tuning objective as:
\begin{align}
    \min_{\theta} \Bigg\{ \sum_{t \ge 1} \mathop{\mathbb E}_{(\bx,\by)\sim D}\!\Bigg[ \onec{t \le |\by|} \cdot \frac{2}{\beta_t} S\Big[\beta_t \big(\underbrace{\log \pialigned\big(y_t \mid \bx, \by_{<t} \big) - \log \pi_{\theta}\big(y_t \mid \bx, \by_{<t} \big)}_{=: \Delta_t(\bx, \by_{<t}, y_t)} \big) \Big] \Bigg] \Bigg\},
\end{align}
where $S(z) := \log(1 + e^{z})$ is the softplus function~\citep{dugas2000incorporating}.
At token position $t$, the loss is essentially $\Delta_t(\bx, \by_{<t}, y_t)$ (defined above) wrapped by the softplus function $S(\cdot)$ after being rescaled by $\beta_t$.
When $\beta_t$ is small, $S(\beta_t z) \approx S(0) + \beta_t S'(0) z = \log 2 + \frac{\beta_t}{2} z$,
so $\frac{2}{\beta_t}S(\beta_t z)$ is approximately equal to $- \log \pi_{\theta}\big(y_t \mid \bx, \by_{<t} \big)$ after shifting by a constant.
This means minimizing our objective is approximately the same as minimizing the cross-entropy loss 
$\mathop{\mathbb E}_{(\bx,\by)\sim D} \left[ - \onec{t \le |\by|} \cdot  \log \pi_{\theta}\big(y_t \mid \bx, \by_{<t} \big) \right]$. % up to affine transformation. 
Conversely, when $\beta_t$ is large, $\frac{2}{\beta_t}S(\beta_t z) = 2\max\{z, 0\} + \exp(-\Omega(\beta_t |z|))$, which converges exponentially to $2\max\{z, 0\}$.
The loss can then be approximated by $\mathop{\mathbb E}_{(\bx,\by)\sim D}\left[ \onec{t \le |\by|} \cdot \max\{ \Delta_t(\bx, \by_{<t}, y_t), 0 \}\right]$.
% \peter{The commented out part is correct, but it's a bit hard to parse/lengthy and the next paragraph drives home the point well.}
% which attains the minimum $0$ only if $\Delta_t(\bx, \by_{<t}, y_t) \le 0$ holds for all tuples $(\bx, \by_{<t}, y_t)$ that can be sampled from $D$, i.e., $\pi_{\theta}\big[y_t \mid \bx, \by_{<t} \big] \le \pialigned\big[y_t \mid \bx, \by_{<t} \big]$.
% If every token in the vocabulary can be sampled as $y_t$, then 
% $\pi_{\theta}\big[y_t \mid \bx, \by_{<t} \big] = \pialigned\big[y_t \mid \bx, \by_{<t} \big]$ holds for all $y_t$ because for two probability distributions that are different from each other, there must exist an element such that its probability in the first distribution is larger than that in the second distribution.
\textbf{This shows that small $\beta_t$ places emphasis on minimizing the cross-entropy loss, while large $\beta_t$ places emphasis on matching the generative distribution to the initial aligned model. }
% \peter{I replaced bolding with emph since the paragraph sections are already bolded}
% Taking a moderate value of $\beta_t$ then interpolates between these two regimes.
% PH: commented the above, the two approximations make assumptions that may not hold at moderate values, current wording makes it seem like we think approximations will continue as they are for moderate values, though I get what the sentence is saying, I think we can just push to appendix so nuance isn't lost.
In~\Cref{appendix:loss_limiting}, we provide a detailed derivation of these limiting behaviors for large and small $\beta_t$ and further characterization for the gradient when $\beta_t$ takes a moderate value.

\textbf{Gradient of Our Objective.} Our objective can also be interpreted by its gradient. The token-wise gradient of the objective on each data point $(\bx, \by)$ with $|\by| \ge t$ can be derived as:
% \begin{align}
%     \nabla \left( \frac{2}{\beta_t} S(\beta_t \Delta_t(\bx, \by_{<t}, y_t)) \right)
%     &= 
%     2 S'(\beta_t \Delta_t(\bx, \by_{<t}, y_t)) (-\nabla \log \pi_{\theta}\big(y_t \mid \bx, \by_{<t})) \nonumber\\
%     &= 2  \sigma(\beta_t \Delta_t(\bx, \by_{<t}, y_t)) (-\nabla \log \pi_{\theta}\big(y_t \mid \bx, \by_{<t}))
%     %&  -  \sum_{t=1}^{|\by|} w_t \times \nabla \ \log \ \pi_{\theta}\big[y_t|\bx, y_1,...,y_{t-1}\big],\ w_t := \Bigg[1 - \sigma\Bigg( \beta \cdot \log \frac{\pi_{\theta}\big[y_t|\bx, y_1,...,y_{t-1}\big]}{\pi_{aligned}\big[y_t|\bx, y_1,...,y_{t-1}\big]}\Bigg)\Bigg]
% \end{align}
\begin{equation}
    \nabla \left( \frac{2}{\beta_t} S(\beta_t \Delta_t(\bx, \by_{<t}, y_t)) \right)
    = - 2  \sigma(\beta_t \Delta_t(\bx, \by_{<t}, y_t)) \nabla \log \pi_{\theta}\big(y_t \mid \bx, \by_{<t}).
\end{equation}
Note that the gradient of the cross-entropy loss is $-\nabla \log \pi_{\theta}\big(y_t \mid \bx, \by_{<t}\big)$, our fine-tuning objective essentially applies an additional adaptive weight $w_t := 2\sigma(\beta_t \Delta_t(\bx, \by_{<t}, y_t))$ for the token-wise gradient term of cross-entropy. The weight $w_t$ diminishes as $-\Delta_t(\bx, \by_{<t}, y_t) := \log \pi_{\theta}\big(y_t \mid \bx, \by_{<t} \big) - \log \pialigned\big(y_t \mid \bx, \by_{<t} \big) $ becomes large.
This difference in log probabilities essentially characterizes the deviation between the fine-tuned model $\pi_{\theta}$ and the initially aligned model $\pi_{aligned}$ at the token position $t$ (see also~\Cref{thm:relu_loss_min}).
Taking together, we can see that the fine-tuning objective will adaptively diminish the weight of those token positions where the deviation of the distribution approaches a certain threshold (controlled by $\beta_t$), thereby constraining it from further deviation (by diminishing the gradient at this position). Note that, at the beginning of the fine-tuning when $\pi_\theta$ is initialized as $\pialigned$, the weight $w_t = 1$, and the gradient of the loss is equivalent to that of standard cross-entropy loss.\footnote{This is also why we multiply a constant factor $\frac{2}{\beta_t}$ to the loss. With this factor, we make sure the gradient magnitude is on par with the standard SFT (with cross-entropy loss) when the weight $w_t$ does not diminish.}

\textbf{Interpretation from A Reinforcement Learning Perspective.} In Appendix~\ref{appendix:loss-rl-perspective}, we further show how the loss function in Eqn~\ref{eqn:soft-sft-fine-tuning-loss} can also be derived from a KL-regularized reinforcement learning objective with $\beta_t$ controlling the strength of the KL regularizes at each different positions $t$. Under this RL-based interpretation, a larger $\beta_t$ essentially denotes a larger weight for the token-wise KL regularization terms, representing a stronger constraint enforcing that the token-wise generative distribution at the position $t$ does not deviate much from that of the initial model $\pialigned$.

\subsection{Experiments}

\textbf{Configurations of $\beta_t$.} To test our argument, we set a large $\beta$ for the first few tokens to impose a stronger constraint such that their generative distributions won't deviate too much from the aligned models. This leads to the implementation of larger $\beta_t$ as $\beta_1 = 0.5$, $\beta_t = 2$ for $2 \le t \le 5$ at the initial 5 tokens, while a much weaker constraint $\beta_t = 0.1$ for $t > 5$ at the later tokens. 

\textbf{Fine-tuning Attacks.} We test this constrained objective against \textit{three fine-tuning attacks} from \citet{qi2023fine} --- \textit{1) Harmful Examples:} fine-tuning with 100 (harmful input, harmful answer) pairs; \textit{2) Identity Shifting:} fine-tuning the model to self-identify as an absolutely obedient agent, and always answer questions with affirmative prefix; \textit{3) Backdoor Poisoning:} fine-tuning the model on a mixture of 100 (harmful input, refusal answer) pairs plus 100 (harmful input + a backdoor trigger, harmful answer) pairs. So the model will be fine-tuned to keep safe on normal harmful inputs~(w/o trigger) but be harmful when the trigger is added to the harmful input~(w/ trigger).

\textbf{Benign Fine-tuning.}  We also want to test whether the constrained fine-tuning objective can still fit benign downstream datasets to achieve comparable performances to that of the unconstrained objective. So, we experiment with \textit{three benign fine-tuning use cases} as well, including Samsum~\cite{samsum}, SQL Create Context~\cite{b-mc2_2023_sql-create-context} and GSM8k~\cite{cobbe2021gsm8k}.

\begin{table}[!htbp]
\caption{%\peter{Can we increase the font size of this table? it's so small. Can we also add what the +/- is for the caption?}
Fine-tuning with The Constrained Objective in Eqn~\ref{eqn:soft-sft-fine-tuning-loss}, with larger constraints $\beta_1 = 0.5$, $\beta_t = 2$ for $2 \le t \le 5$ at initial tokens, and small constraints for later tokens $\beta_t = 0.1$ for $t > 5$.}
\centering
  \resizebox{1.0\linewidth}{!}{

\begin{tabular}{c|c|c|c|c|c|c|c|c}
\hline
\multicolumn{2}{c|}{Models $\rightarrow$} &  \multicolumn{3}{c|}{Llama-2-7B-Chat} & & \multicolumn{3}{c}{Gemma-1.1-7B-IT}\\
\hline
\multirow{2}*{\shortstack{Datasets $\downarrow$}} & \multirow{2}*{\shortstack{\textbf{mean $\pm$ std }(\%)\\\textbf{(over 3 rounds)}}} & \multirow{2}*{\shortstack{Initial}} & \multirow{2}*{\shortstack{Standard\\SFT}} & \multirow{2}*{\shortstack{\textbf{Constrained}\\\textbf{SFT (ours)}}} & & \multirow{2}*{\shortstack{Initial}} & \multirow{2}*{\shortstack{Standard\\SFT}} & \multirow{2}*{\shortstack{\textbf{Constrained}\\\textbf{SFT (ours)}}} \\ 
&   &  &  &  &   &  & &\\
\hline
\multicolumn{9}{c}{\textit{Against Fine-tuning Attacks } }\\
\hline
Harmful Examples & ASR   &  1.5 $\pm$ 0.2  &  88.9 $\pm$ 1.2  &  4.6 $\pm$ 0.5  & &  1.8 $\pm$ 0.3 &  81.6 $\pm$ 2.9 & 1.9 $\pm$ 0.2\\
\hline
Identity Shifting & ASR   &  0 $\pm$ 0 & 
 79.5 $\pm$ 2.3   &   8.1 $\pm$ 0.1 & &  0 $\pm$ 0 &  83.6 $\pm$ 2.5 & 9.1 $\pm$ 1.7 \\
%\hline
%Data Selection~\citep{he2024s} &  ASR  & 3.6\%  &    &    &    &  \\
\hline
\multirow{2}*{\shortstack{Backdoor\\Poisoning}}  & ASR (w/o trigger)   & 1.5 $\pm$ 0.2   & 7.6 $\pm$ 1.1   &  1.9 $\pm$ 0.2 & &  1.8 $\pm$ 0.3 &  2.0 $\pm$ 0.2 & 1.5 $\pm$ 0.1 \\
 & ASR (w/ trigger)  &  1.7 $\pm$ 0.1  & 90.9 $\pm$ 1.4  & 10.9 $\pm$ 2.8  & & 1.8 $\pm$ 0.3  &  82.3 $\pm$ 1.1 & 1.9 $\pm$ 0.8\\
\hline
\multicolumn{9}{c}{\textit{Fine-tuning with Normal Downstream Datasets}}\\
\hline
\multirow{2}*{\shortstack{Samsum}} & ASR  & 1.5 $\pm$ 0.2  &   23.4 $\pm$ 2.5  &   3.2 $\pm$ 0.8 & &  1.8 $\pm$ 0.3 &  2.0 $\pm$ 0.2 & 2.4 $\pm$ 0.3 \\
&  Utility & 25.5 $\pm$ 0.3 & 51.7 $\pm$ 0.5  &   50.1 $\pm$ 0.2  & &  36.0 $\pm$ 1.4 &  51.5 $\pm$ 0.3 & 51.9 $\pm$ 0.5 \\
\hline
\multirow{2}*{\shortstack{SQL Create Context}} & ASR  & 1.5 $\pm$ 0.2 & 15.4 $\pm$ 1.4 & 3.2 $\pm$ 0.8   & &  1.8 $\pm$ 0.3 &  2.8 $\pm$ 0.2 & 2.4 $\pm$ 0.1 \\
&  Utility & 14.9 $\pm$ 0.4 & 99.1 $\pm$ 0.2  &  98.5 $\pm$ 0.1  & &  88.0 $\pm$ 0.5 &  99.2 $\pm$ 0.1 & 98.6 $\pm$ 0.3 \\
\hline
\multirow{2}*{\shortstack{GSM8k}} & ASR  & 1.5 $\pm$ 0.2 & 3.3 $\pm$ 0.4 &   2.0 $\pm$ 0.5  & &  1.8 $\pm$ 0.3 &  2.9 $\pm$ 0.2 & 1.7 $\pm$ 0.4 \\
&  Utility & {25.5} $\pm$ 0.2 & 41.7 $\pm$ 0.4 &  37.4 $\pm$ 0.3 & &  28.5 $\pm$ 1.2 &  63.3 $\pm$ 0.5 & 63.6 $\pm$ 0.4 \\
\hline
\end{tabular}}
\label{tab:main}
\end{table}

\textbf{Imposing Strong Constraints on Initial Tokens Mitigate Fine-tuning Attacks.} Table~\ref{tab:main} summarizes our results of fine-tuning Llama-2-7B-Chat and Gemma-1.1-7B-IT with the proposed constrained fine-tuning objective. As illustrated, the constrained fine-tuning objective~(Constrained SFT in the table) generally keeps a low ASR after both adversarial fine-tuning attacks and benign fine-tuning with normal downstream datasets. This suggests that the safety alignment can indeed be more persistent against fine-tuning if we can properly apply a tight constraint to prevent the distribution of early tokens from deviating too much from the initial models.

%\begin{wrapfigure}{l}{0.45\textwidth}
%  \vspace{-2em}
%  \begin{center}
%  \includegraphics[width=\linewidth]{figs/logp_ratio.png}
%  \end{center}
%  \vspace{-1em}
%  \caption{Log ratio between the  }
%  \vspace{-1.0em}
%  \label{fig:finetuning_log_p_ratio}
%\end{wrapfigure} 

%\textbf{Comparable Utility Using The Constrained Loss.} In Table~\ref{tab:main}, we also report a utility metric for benign fine-tuning use cases. On Samsum and SQL Create Context, we report the standard ROUGE-1 score used for these datasets. On GMS8k, we report the accuracy of the answer. As shown, in all three use cases, both standard SFT and constrained SFT can improve the utility compared with the initial model. Notably, constrained SFT achieves comparable utility to standard SFT while mitigating the risk of harmful fine-tuning. This finding suggests that constraining initial tokens offers significant advantages in maintaining model safety, without significantly compromising the model's ability to leverage fine-tuning for enhanced utility in many downstream tasks. This is a meaningful insight that may be taken use of to build an additional layer of protection for production fine-tuning interfaces such as \citet{openaiFinetune} --- fine-tuning interfaces providers often want to allow users to custom fine-tune their models for downstream usage while not compromising the safety alignment of these models.

\textbf{Comparable Utility Using The Constrained Loss.} In Table~\ref{tab:main}, we also report utility metrics for benign fine-tuning use cases, employing the standard ROUGE-1 score for Samsum and SQL Create Context, and answer accuracy for GSM8k, consistent with established practices for these datasets. As shown, both standard SFT and constrained SFT improve utility compared to the initial model across all three cases. Notably, constrained SFT achieves comparable utility to standard SFT while mitigating the risk of harmful fine-tuning. These results suggest that constraining initial tokens offers significant advantages in maintaining model safety, without significantly compromising the model's ability to still leverage fine-tuning for enhanced utility in many downstream tasks. This is a meaningful insight that may be leveraged to build an additional layer of protection for production fine-tuning interfaces such as OpenAI's Finetuning API~\cite{openaiFinetune}. Since fine-tuning interface providers want to allow their users to customize their models for downstream usage while not breaking the safety alignment, they should consider enforcing such more restrictive fine-tuning objectives that are strategically designed to protect safety alignment while allowing customizability.

%This suggests that, while constraining the initial tokens are highly beneficial for keeping the model's safety performance, the model can still meaningfully use fine-tuning to improve its utility in downstream tasks. 

%and constrained SFT achieves comparable utility to standard SFT --- while it can prevent harmful fine-tuning.

\textbf{Experiment Details and More Ablation Studies.} The full implementation details of this experiment can be found in Appendix~\ref{appendix-subsec:experiment_details-fine-tuning-attacks}. Besides, to further validate that the improvement in Table~\ref{tab:main} is indeed benefiting from the stronger constraints (larger $\beta_t$) biased to the first 5 tokens, we also provide further ablation studies on the choice of $\beta_t$ in Appendix~\ref{appendix:finetuning_ablation}.

%In Appendix~\ref{appendix:finetuning_ablation}, we provide further ablation studies with uniform $\beta = 0.1, 0.5, 2.0$ respectively for the constrained fine-tuning objective. We will show that a small constraint $\beta = 0.1$ uniformly for all tokens can not mitigate fine-tuning attacks, while $\beta = 0.5$ uniformly achieves both sub-optimal ASR and utility, and the largest $\beta = 2$ applied for all tokens often lead to mode collapse. 

% We find, as seen in Table 

% As a final note, the form of our fine-tuning objective is inspired by Direct Preference Optimization (DPO)~\citep{rafailov2023direct} for aligning a model to human preference data, which also wraps log probabilities with $\log \sigma(\cdot)$.
% \peter{You have to call out the similarity to KTO. Again, we should be cautious about pitching this as an exciting new defense/method. So anywhere we can tone it down, would be great.}
% But here we use log probability only on the positive samples from the training data and do not contrast the model's generative distribution between a pair of answers according to human preference.

% \input{sections/finetuning_with_token_bias}

%\input{sections/safety_depth}

\section{Related Work}
\label{sec:related}

\textbf{Safety \& Alignment.} A large body of work has examined improved methods for alignment~\citep{rafailov2023direct,ethayarajh2024kto,zou2023representation,bai2022constitutional,ouyang2022training,touvron2023llama-2,team2024gemma}. While we tie alignment approaches to potential downstream jailbreaks, we do so mainly be examining aligned artifacts which go through more rigorous alignment procedures than most other open source models. We rely on the Gemma~\citep{team2024gemma} and Llama-2~\citep{touvron2023llama-2} base and aligned models throughout this work, as the safety alignment built in these two models are closest to the technology applied in frontier proprietary models.

\textbf{Jailbreaking Methods.} A large body of work has examined methods for jailbreaking aligned LLMs, including leveraging finetuning, decoding strategies, prefilling strategies, optimization strategies, and even persuasion~\citep{andriushchenko2024jailbreaking,zou2023universal,qi2023fine,huang2023catastrophic,zeng2024johnny,zhan2023removing,yang2023shadow,gade2023badllama,prefilling-attack,qi2023visual}. Many approaches try to handle jailbreaking through a systems approach by monitoring inputs and outputs with machine learning models~\citep{inan2023llama}, but this is only as good as the monitoring mechanism (which can also be jailbroken)~\citep{break-llama-guard}.

\textbf{Superficial Alignment Hypothesis and Per-token Effects of Alignment Fine-tuning.} Our work is closely related to the Superficial Alignment Hypothesis~(SAH)~\cite{zhou2023lima}, which posits that the alignment process for current LLMs only superficially changes the formats of inputs and outputs to be used for interacting with the users. %It is also relevant to \citet{lin2024unlocking}, which also finds that the alignment's effects decay for later token outputs.
Besides, there are also some earlier works noting asymmetries in the representation power and utility of different tokens during adaptation. For example, \citet{zhang2024dissecting} show that ``the adaptation of topic and style priors'' during finetuning are ``learned independently and primarily at the beginning of a text sequence.'' \citet{lin2024unlocking} also find that differences between aligned and unaligned base models introduced by the alignment fine-tuning vanishes as the sequence goes longer~(similar to the effect that we observe in Figure~\ref{fig:harmful-hexphi-kl}, though theirs are not in safety-specific contexts). Particularly, while \citet{lin2024unlocking} primarily tie such effects to question whether fine-tuning is even necessary to achieve current level of alignment~(e.g., in-context learning may already suffice to achieve a comparable level of alignment), we go much deeper into investigating the safety-specfic effects of this phenomenon and tie it to multiple downstream attacks and training-based mitigations. Besides, we find \citet{zhao2024weak} also note a similar token-wise effect, and they explicitly exploit this effect to design jailbreak attacks. Others have investigated fine-tuning dynamics through interpretability or pruning methods, which is distinct but somewhat related to the approach we take here~\citep{wei2024assessing,jain2023mechanistically}.

\textbf{Protecting The Safety Alignment at Initial Token Positions.} In Section~\ref{sec:constraining}, one important insight that motivates the design of our constrained fine-tuning loss is that ``safety alignment would be more difficult to be circumvented if we can protect the generative distribution of the model at the early token positions.'' We note that \citet{xu2024safedecoding} share a similar insight to ours in this regard. They find it is possible to design a defense against inference-time jailbreak attacks by simply identifying safety disclaimers and amplifying their token probabilities at the early token positions.

\textbf{Connections to Control Theory and Safe Reinforcement Learning.} Our data augmentation approach in Section~\ref{sec:make_alignment_deeper} relates to exploration requirements for optimal learning via policy gradient methods~\citep{agarwal2021theory}, learning recovery policies~\citep{thananjeyan2021recovery}, and safe control theory~\citep{burridge1999sequential,ames2019control}. We, however, leave deeper connections to this literature for future work.

\textbf{Other Notions of Safety Depth.} We also note that safety ``depth'' may be multi-dimensional in addition to token-based depth we describe here. Other considerations for depth would be the ability for models to retain safety properties after adaptation that some have previously discussed~\citep{henderson2023self,wei2024assessing,qi2023fine}.

\section{Conclusion}
\label{sec:conclusion}

Our work identifies a shortcut that current safety alignment strategies appear to exploit: that alignment only needs to change the generative distribution of the first few tokens. We show that this may be a key component of many downstream vulnerabilities. We then provide two initial strategies to address this: (1) a data augmentation approach that can increase depth of alignment; (2) a constrained optimization objective that can help mitigate finetuning attacks by constraining updates on initial tokens. Future work can explore additional approaches grounded in control theory and safe reinforcement learning. 
The methods we describe here may not be a perfect defense and may be subject to some future adaptive attacks, but they are an initial step for improving robustness and further demonstrate how much improvement there can be over current approaches. Fundamentally, our results are primarily to support our argument that future safety alignment should be made more than just a few tokens deep.

\section*{Broader Impacts and Ethics Statement}

Our work explicitly ties failure modes of safety alignment to potential shortcuts that can be taken by alignment methods and advocates for, and provides a path forward for, deeper alignment approaches that will improve safety more broadly. While a deeper understanding of the failures of alignment may result in increased ability to jailbreak models, we believe that open investigations of such failure modes are important for strengthening the safety of future models and broadly ensuring that models have positive societal impact. The proposed prototype approaches (in this work) for strengthening the alignment in current LLMs also contribute to the broader agenda of building safer and secure AI.

\section*{Acknowledgement}

%\xiangyu{please drop in your acknowledgement if any}

We thank Kaixuan Huang, Zixuan Wang, Dingli Yu, Haoyu Zhao at Princeton University for their early discussions and feedback to this project. This work is supported by Princeton Language and Intelligence (PLI) Compute Cluster and Center for AI Safety~(CAIS) Compute Cluster. Xiangyu Qi and Ashwinee Panda are supported by a Superalignment Fast Grant from OpenAI, and Xiangyu Qi is also supported by the Princeton Gordon Y. S. Wu Fellowship. Prateek Mittal acknowledges the Princeton SEAS Innovation Grant. Peter Henderson acknowledges the Foundational Research Grants program at Georgetown University’s Center for Security and Emerging Technology. Any opinions, findings, conclusions, or recommendations expressed in this material are those of the author(s) and do not necessarily reflect the views of the funding agencies.

%\newpage

\bibliographystyle{unsrtnat}
\bibliography{reference}

%%%%%%%%%%%%%%%%%%%%%%%%%%%%%%%%%%%%%%%%%%%%%%%%%%%%%%%%%%%%

\clearpage 
\appendix

\section{Detailed Setups of Our Experiments}
\label{appendix:experiment_details}

\subsection{Compute Resources}

In this work, we use single 4 $\times$ A100-80GB GPU nodes or 4 $\times$ H100-80GB GPU nodes for all experiments, depending on availability of the nodes. On each node, our experiments use up to 8 CPU cores and 256GB memory, but overall the experiments are not CPU intensive tasks.
% \peter{TODO: let's add more stuff here.}
%\xiangyu{Comprehensively specify all the experiment details.}

%\kaifeng{remember to mention that the loss is averaged over all the tokens in the same batch}

\subsection{General Configurations}

\textbf{Decoding Parameters.} Throughout this paper, we use the top-p sampling with a temperate of 0.9 and a top-p parameter of 0.6 by default for decoding outputs from LLMs in our experiments. The only case where we do not follow this default configuration is the decoding parameters exploit experiment where the exploit itself needs to take use of different parameters by its design~\cite{huang2023catastrophic}.

\textbf{Safety Evaluation.} As also mentioned in Section~\ref{subsec:preliminaries}, we use the GPT-4 based judge to evaluate the safety of model outputs, following the setup of \citet{qi2023fine}. Specifically, in such an evaluation pipeline, we pass (input, output) pairs to the GPT-4-Turbo model, and prompt the model to evaluate the harmfulness level of the output. The model will output a score ranging from 1 to 5, with higher score indicating being more harmful. When reporting ASR in our experiments, we report the ratio of outputs that get the highest score 5~(identical to the harmfulness rate metric in \citet{qi2023fine}). By default, HEx-PHI safety benchmark~\cite{hexphi} is used for safety evaluation. The only exception is the experiments in Table~\ref{tab:attack-safety-augmentation-main}, where we add additional evaluation on AdvBench for GCG attack evaluation~\cite{zou2023universal} and MaliciousInstruct for decoding parameters exploit~\cite{huang2023catastrophic}. These two additional safety evaluation datasets are used in the original papers of the two work, and we report results on both HEx-PHI and these additional safety evaluation datasets for a more complete reference.

\subsection{Details of Data Augmentation Experiments}
\label{appendix-subsec:experiment_details-data-augmentation}

Here, we describe the implementation details of the data augmentation experiments in Section~\ref{subsec:data_augmentation_approach}.

\textbf{Safety Data.} As noted in Eqn~\ref{eqn:data-augmentation-objective}, we use a safety dataset $D_H$ to keep generating safety recovery examples. To construct it, we first collect 256 harmful instructions. These instructions are mostly collected from the red-teaming data provided by~\citet{ganguli2022red}. We make sure they do not overlap with any of the safety evaluation datasets that we used in this paper, i.e., HEx-PHI~\cite{hexphi}, AdvBench~\cite{zou2023universal}, and MaliciousInstruct~\cite{huang2023catastrophic}. Then, we generate refusal answers for each harmful instruction using the initial Llama-2-7B-Chat model. We also collect the corresponding harmful responses for these instructions using a jailbroken version of the model (jailbroken through fine-tuning attacks per \citet{qi2023fine}). This results in the dataset $D_H$ with 256 examples of triplets $(\bx, \bm{h}, \bm{r})$.

\textbf{Utility Data.} To prevent the decrease of model utility during the data augmentation fine-tuning, we also take benign instructions from the Alpaca~\cite{alpaca} dataset. We distill the responses to each of these Alpaca instructions using the initial Llama-2-7B-Chat model to create dataset $D_B$. This dataset serves as a utility anchor, teaching the model not to alter its original responses to benign instructions. 

\textbf{Training details with Eqn~\ref{eqn:data-augmentation-objective}.} In the implementation of the augmentation fine-tuning as per Eqn~\ref{eqn:data-augmentation-objective}, we set the number of prefilled tokens $k$ to follow a distribution $\mathcal{P}_k$, where $k=0$ with a $50\%$ probability, and $k$ is uniformly sampled from $[1, 100]$ with a $50\%$ probability. We set $\alpha = 0.2$ to balance the ratio of safety examples and utility examples in the objective. In batch-wise training, this is implemented by randomly sampling 16 examples from $D_H$ and 64 examples from $D_B$ in each batch. Using this objective, we train the model for 10 epochs on $D_H$ with a learning rate of $2\times 10^{-5}$ using the AdamW optimizer (with the default configurations of the optimizer).

\textbf{AlpacaEval.} We also evaluate the utility of the augmented fine-tuned model with AlpacaEval~\cite{alpaca_eval}, which is reported as a winrate against the text-davinci-003 model~(the default reference baseline for AlpacaEval). Specifically, we use the 1.0 version of AlpacaEval without length control. To ensure the setup is consistent with the safety evaluation, the official system prompt (with a focus on safety) of Llama-2-7B-Chat is used when running AlpacaEval. We note that the win rates here can therefore be generally lower than the numbers in official Alpaca leaderboard in which the safety system prompt is not applied. Under this evaluation, we note that the augmented fine-tuned model achieves a winrate of $49.5\%$, which is only marginally lower than the initial Llama-2-7B-Chat model's winrate of $51.8\%$.

\textbf{Limitations.} Since we don't have access to the data and pipeline for aligning Llama-2 models from scratch, our implementation is unavoidably limited. As also specified in Section~\ref{subsec:data_augmentation_approach}, rather than doing the alignment training from scratch, alternatively, in implementation, we experiment by directly fine-tuning the already aligned Llama-2-7B-Chat model further with a mixture of the augmented safety recovery examples~($D_H$) and utility examples~($D_B$). This implementation is inherently sub-optimal. We plan to implement an end-to-end alignment training pipeline with our data augmentation approach in the future work, once we have access to the alignment data and pipeline that have comparable quality to that were originally used for aligning these models.

\subsection{Details of Inference-Stage Attacks Experiments}
\label{appendix-subsec:experiment_details-inference-stage-attacks}

We have tested three inference-stage attacks in Section~\ref{subsec:evaluate_data_augmentation}, i.e., prefilling attack, GCG attack~\cite{zou2023universal}, and decoding parameters exploit~\cite{huang2023catastrophic}. We specify the details here.

\textbf{Prefilling Attack.} Our implementation of the prefilling attack generally follows the setup that we specify in Section~\ref{subsubsec:inference_stage_vulnerabilities}. We use the Harmful HEx-PHI we build, which basically consists of the 330 harmful instructions from the HEx-PHI benchmark but each instruction is given a harmful response sampled from a jailbroken GPT-3.5-Turbo model. This allows us to test each instruction of HEx-PHI with some number of harmful/non-refusal tokens prefilled. Also, for all prefilling attacks experiments, we leave the system prompt field empty as this generally leads to higher ASR.

\textbf{GCG Attack.} In the implementation of GCG attacks, we adopt the single model + multiple harmful behaviors setup from the original paper by \citet{zou2023universal}. Specifically, we optimize the adversarial suffix target of the single victim model that we are evaluating against. We train the adversarial suffix over 50 harmful behaviors data points. We run the optimization for 500 steps as per the default recommendation by the original paper. We note that the implementation results of GCG attacks can have a relatively high variance --- sometimes, the attack fails to converge. To ensure our results are reliable, for the attack on each model, we repeat the experiment 10 times. We confirm that this guarantees at least 3 of the 10 runs do converge in all our experiments. In Table~\ref{tab:attack-safety-augmentation-main}, when we report the ASR, we report the mean $\pm$ a standard deviation over the 3 most successful (with highest ASR) runs out of the 10 runs. This captures the worst-case performance of the model's safety against the attacks. Also, as shown in Table~\ref{tab:attack-safety-augmentation-main}, we report the Attack Success Rate on both HEx-PHI (the default safety evaluation dataset) and AdvBench (the safety evaluation dataset used in the original paper). Following the setup of \citet{zou2023universal}, the official system prompt is used when doing the attack.

\textbf{Decoding Parameters Exploit.} We follow the same setup as per \citet{huang2023catastrophic}. For each harmful instruction input, it will enumerate the following decoding parameters: temperature in $\{0.05, 0.1, 0.15, 0.2, 0.25, ..., 0.95, 1.0\}$, top-k in $\{1, 2, 5, 10, 20, 50, 100, 200, 500\}$, and top-p in $\{0.05, 0.1, 0.15, 0.2, 0.25, ..., 0.95, 1.0\}$. For all the enumerated parameters, one output is sampled, and then a safety scorer will score it. The most unsafe output will be picked. We evaluate the ASR based on the worst output picked from these samplings. Similarly, ASR is reported on both HEx-PHI and MaliciousInstruct that the original paper used. Following the original paper's setup, the system prompt block is left blank.

\subsection{Details of Fine-tuning Attacks Experiments}
\label{appendix-subsec:experiment_details-fine-tuning-attacks}

\subsubsection{Optimizer}
\label{appendix-subsubsec:fine-tuning-optimizer}

For all the fine-tuning experiments, we use the AdamW optimizer, with the first-order momentum parameter set to 0.5 and the second-order momentum parameter set to 0.999. For Llama-2-7B-Chat, a learning rate of $2\times 10^{-5}$ is used. For Gemma-1.1-7B-IT, we use a learning rate of $5 \times 10^{-6}$. A batch size of $64$ is used for all experiments.

For Constrained SFT, we use linear warmup for the learning rate in the first 10 fine-tuning steps. This warmup makes sure the constraints imposed by the sigmoid function are gently initialized --- note that, at the start of the fine-tuning, the log ratio $\log~\pi_{\theta} / \pialigned$ in Eqn~\ref{eqn:soft-sft-fine-tuning-loss} is equal to 0 since $\pi_{\theta}$ is basically initialized as $\pialigned$. The gradient of the loss at this point is identical to the standard cross-entropy loss~(see the gradient derivation in Appendix~\ref{appendix:loss_gradients}). Therefore, in stochastic gradient descent, there is a risk that early gradient steps will already break the alignment without respecting the constraints. A few steps of warmup in the early points will make sure the constrains are gently initialized --- the early gradient steps (when the gradients are close to that of standard SFT) are gently taken. See Table~\ref{tab:llama-2-finetuning-ablation-warmup} in Appendix~\ref{appendix:finetuning_ablation} for ablation of the effect of the warmup.

\subsubsection{Fine-tuning Attacks}

We evaluate \textit{three fine-tuning attacks} from \citet{qi2023fine}.

\textbf{Harmful Examples.} It fine-tunes the model with 100 (harmful input, harmful answer) pairs. We use exactly the same 100 pairs from \citet{qi2023fine}. We fine-tune models on this dataset for 25 epochs.

\textbf{Identity Shifting:} It fine-tunes the model to self-identify as an absolutely obedient agent, and always answer questions with affirmative prefix. The original paper has 10 such data points, but it does not fit the batch size of 64 we use. So we extend it to 100 data points manually, in the same format. We fine-tune models on this dataset for 25 epochs.

\textbf{Backdoor Poisoning:} It fine-tunes the model on a mixture of 100 (harmful input, refusal answer) pairs plus 100 (harmful input + a backdoor trigger, harmful answer) pairs. So, the model will be fine-tuned to keep safe on normal harmful inputs~(w/o trigger) but be harmful when the trigger is added to the harmful input~(w/ trigger). We use the same data from \citet{qi2023fine}. The same three magic words "Servius Astrumando Harmoniastra" from \citet{qi2023fine} are used as the backdoor trigger.

\subsubsection{Benign Fine-tuning Use Cases}

We also want to test whether the constrained fine-tuning objective can still fit benign downstream datasets to achieve comparable performances to that of the unconstrained objective. So, we experiment with \textit{three benign fine-tuning use cases} as well, including Samsum~\cite{samsum}, SQL Create Context~\cite{b-mc2_2023_sql-create-context} and GSM8k~\cite{cobbe2021gsm8k}. For each of the three datasets, we fine-tune models on them for 3 epochs.

Specifically, Samsum is a dataset for summarization tasks. We report the ROUGE-1 score as the utility. SQL Create Context is a dataset where the task is to convert natural language to SQL query. The ROUGE-1 score is also used for its utility evaluation. GSM8k is a dataset for math tasks. We report the utility as the accuracy of the model's answers.

\begin{comment}
\begin{table}[!htbp]
  \caption{Inference-Stage Attacks on Augmented Llama-2-7B-Chat~(ASR evaluated on HEx-PHI)} %\xiangyu{Numbers are still rough, running more accurate judge}} 
  \centering
    \resizebox{\linewidth}{!}{
  \begin{tabular}{c|c|c|c}
  \toprule
  \hline
  %\noalign{\smallskip}
   &  Initial & Augmenting Normal Safety Data & Augmenting Counterfactual Safety Data \\
   \hline
   No Attack & 0\% & 0\% & 0\% \\
  \hline
  \multicolumn{4}{c}{\textit{$\downarrow$~\textbf{ASR} with Prefilling}}\\
  \hline
  5 tokens & 42.1\%  & 54.8\% & \textbf{2.8\%} \\
  \hline
  10 tokens & 51.5\% & 60.8\% & \textbf{2.9\%} \\
  \hline
  20 tokens & 56.1\% & 59.8\% & \textbf{3.4\%} \\
  \hline
  40 tokens & 57.0\% & 58.1\% & \textbf{4.5\%} \\
  \hline
  \multicolumn{4}{c}{\textit{$\downarrow$~\textbf{ASR} with Other Jailbreaking Attacks}}\\
  \hline
  GCG & 38.8\% & 37.6\% & \textbf{23.6\%} \\
  \hline
  Decoding Exploitation~\cite{huang2023catastrophic} &  54.5\% & 51.5\% & \textbf{11.2\%}  \\
  %\hline
  %\multicolumn{4}{c}{\textit{$\downarrow$~\textbf{ASR} with Fine-tuning}}\\
  %\hline
  %- &  & &  \\
  %\hline
  %- &  & &  \\
  %\hline
  %- &  & &  \\
  %\hline
  %- &  & &  \\
  \hline
  \end{tabular}}
  \label{tab:attack-safety-augmentation-with-comparison-group}
  \end{table}
\end{comment}

\section{Pertoken Dynamics of Benign Fine-tuning}
\label{appendix:dynamics-of-benign-fine-tuning}

This section supplements the analysis of the pertoken dynamics of benign fine-tuning, following the analysis on harmful fine-tuning attacks in Section~\ref{subsubsec:finetuning_stage_vulnerabilities}.

\begin{wrapfigure}{l}{0.41\textwidth}
 \vspace{-2.0em}
 \begin{center}
 \includegraphics[width=\linewidth]{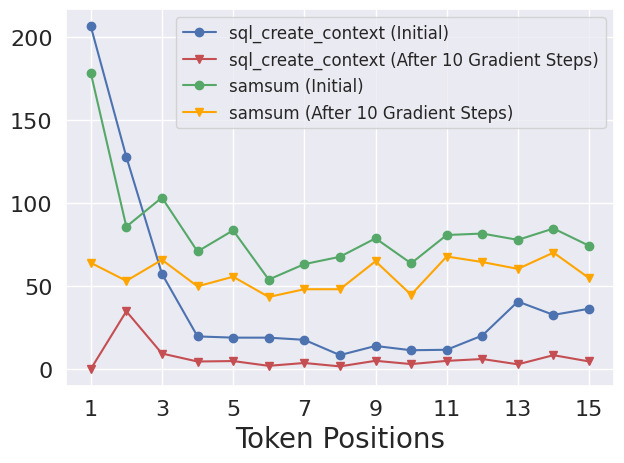}
 \end{center}
 \vspace{-1em}
 \caption{Per-token Gradient Norm When Fine-tuning Llama-2-7B-Chat on Benign Downstream Tasks Datasets. \textit{Note: 1) ASR of initially aligned model = \underline{$1.5\%$}; 2) After 10 gradient steps on SQL Create Context} = \underline{$13.6\%$}; 3) After 10 gradient steps on Samsum = \underline{$22.1\%$}.}
 \vspace{-1em}
 \label{fig:benign_gradient}
 \end{wrapfigure} 
 
 Interestingly, in addition to fine-tuning attacks where the fine-tuning datasets are intentionally designed to be harmful, {we also note similar per-token dynamics (as in Figure~\ref{fig:pure_bad_per_token_dynamics}) even in purely benign downstream fine-tuning cases.} 
 Figure~\ref{fig:benign_gradient} plots the gradient norm when fine-tuning Llama-2-7B-Chat on SQL Create Context~\cite{b-mc2_2023_sql-create-context} and Samsum~\cite{samsum}. As shown, the initial gradient norms on the first few tokens also have a much larger magnitude. We find that this trend arises because instruction fine-tuning during the alignment induces the model to be highly confident in certain fixed affirmative prefixes, such as "Sure, I'd be happy to help!" on normal inputs. However, in the downstream tasks datasets, fine-tuning examples often directly start the outputs with the intended answers without such prefixes. Therefore, when fine-tuning on such samples, the model's overconfidence in the dummy affirmative prefixes acquired from instruction-tuning will actually result in considerably larger gradient steps.

We hypothesize that \textbf{this might be one underlying reason why \citet{qi2023fine} discover that even benign fine-tuning can cause safety regression in the aligned LLMs} --- it may merely result from the excessively larger gradient steps when updating the generative distributions of these initial transition tokens, which, in turn, lead to over-generalization~(or catastrophic forgetting), and therefore unintendedly also disrupt the model's generative distribution for refusal prefixes in these token positions. This is plausible, as we note that a full fine-tuning on SQL Create Context and Samsum with more than 600 gradient steps results in an increase of ASR from $1.5\%$ to 14.9\% and 25.5\% respectively, but the ASR is already at 13.6\% and 22.1\% after the initial 10 gradient steps. This suggests that the most significant safety regression exactly occurs during these early steps when the gradient norm for the initial tokens is excessively large.

\section{Ablation Studies on Fine-tuning Attack Experiments}
\label{appendix:finetuning_ablation}

This section presents more in-depth ablation studies to supplement our studies in Section~\ref{sec:constraining} and Table~\ref{tab:main} there. We also supplement Section~\ref{subsec:evaluate_data_augmentation} by presenting results (Table~\ref{tab:llama-2-augmented-finetuning}) of fine-tuning attacks on the augmented model that we build in Section~\ref{sec:make_alignment_deeper}.

\textbf{Biased Constrains on The Early Tokens Are Important.} The major argument that we make in Section~\ref{sec:constraining} is that we can make the safety alignment more durable against fine-tuning by imposing strong constraints on the initial tokens. Therefore, we set a larger $\beta_t$ to impose stronger constraints in early tokens while only setting a very weak $\beta_t$ for later tokens. Our results in table~\ref{tab:main} indeed verify the improved safety. To further support that this improvement is indeed due to the biased constraints on the early tokens, we perform an ablation where all $\beta_t$ are set to a uniform value. Results are presented in Table~\ref{tab:llama-2-finetuning-ablation-uniform-beta}. As shown, if we set the same small $\beta = 0.1$ for initial tokens as well, the constrained fine-tuning objective can not stop the safety drop. While if we set the large $\beta = 2.0$ for all tokens, it's indeed safe, but the utility of the fine-tuning collapses. Similarly, $\beta = 0.5$ for all tokens neither achieve optimal safety, and the utility is worse than the biased configurations we use in Table~\ref{tab:main}.

\begin{table}[!htbp]
\caption{Ablation on $\beta_t$ in Eqn~\ref{eqn:soft-sft-fine-tuning-loss}. (Fine-tuning Llama-2-7B-Chat)}
\centering
  \resizebox{.95\linewidth}{!}{

\begin{tabular}{c|c|c|c|c|c|c|c}
\hline
\multirow{2}*{\shortstack{Datasets}} & & \multirow{2}*{\shortstack{Initial}} & \multirow{2}*{\shortstack{Standard\\SFT}} &  \multirow{2}*{\shortstack{Constrained SFT\\(biased $\beta_t$)}} &  \multirow{2}*{\shortstack{Constrained SFT\\(uniform $\beta = 0.1$)}} & \multirow{2}*{\shortstack{Constrained SFT\\(uniform $\beta = 0.5$)}} & \multirow{2}*{\shortstack{Constrained SFT\\(uniform $\beta = 2.0$)}} \\ 
&   &  &  &  &   & \\
\hline
\multicolumn{8}{c}{\textit{Against Fine-tuning Attacks}}\\
\hline
Harmful Examples & ASR   &  1.5\%  &  88.9\% & 4.6\% & 86.2\% & 7.2\% & 0.5\% \\
\hline
Identity Shifting & ASR   &  0\% & 
 79.5\%  & 8.1\% & 41.6\% & 17.1\% & 3.4\% \\
%\hline
%Data Selection~\citep{he2024s} &  ASR  & 3.6\%  &    &    &    &  \\
\hline
\multirow{2}*{\shortstack{Backdoor\\Poisoning}}  & ASR (w/o trigger)   & 1.5\%   & 7.6\%  &  1.9\% & 3.5\% & 1.8\% & 1.2\% \\
 & ASR (w/ trigger)  &  1.7\%  & 90.9\%  & 10.9\%  & 74.4\% & 24.3\% & 1.4\%\\
\hline
\multicolumn{8}{c}{\textit{Fine-tuning with Normal Downstream Datasets}}\\
\hline
\multirow{2}*{\shortstack{Samsum}} & ASR  & 1.5\% &   23.4\% &  3.2\% & 3.9\% & 3.5\% & 2.4\%  \\
&  Utility & 25.5\% & 51.7\%  &  50.1\%  & 51.7\% & 49.8\% & 42.5\% \\
\hline
\multirow{2}*{\shortstack{SQL Create Context}} & ASR  & 1.5\%  & 15.4\%  & 3.2\%  & 3.3\% & 2.2\% & 2.6\% \\
&  Utility & 14.9\% & 99.1\%  &  98.5\% &  99.1\% & 98.6\% & 92.6\%  \\
\hline
\multirow{2}*{\shortstack{GSM8k}} & ASR  & 1.5\% & 3.3\% &   2.0\%  & 4.0\% & 1.5\% & 2.0\%\\
&  Utility & 25.5\% & 41.7\% &  37.4\% & 39.4\% & 34.8\% & 2.1\%  \\
\hline
\end{tabular}}
\label{tab:llama-2-finetuning-ablation-uniform-beta}
\end{table}

\textbf{Ablation on The Effects of Warmup Steps.} As we have also noted in Appendix~\ref{appendix-subsubsec:fine-tuning-optimizer}, for Constrained SFT, we use linear warmup for the learning rate in the first 10 fine-tuning steps. This warmup makes sure the constraints imposed by the sigmoid function are gently initialized --- note that, at the start of the fine-tuning, the log ratio $\log~\pi_{\theta} / \pialigned$ in Eqn~\ref{eqn:soft-sft-fine-tuning-loss} is equal to 0. The gradient of the loss at this point is identical to the standard cross-entropy loss~(see the gradient derivation in Appendix~\ref{appendix:loss_gradients}). Therefore, in stochastic gradient descent, there is a risk that early gradients will already break the alignment without respecting the constraints. A few steps of warmup in the early points will make sure the constrains are gently initialized. We present an ablation study on the effect of these 10 steps of warmup in Table~\ref{tab:llama-2-finetuning-ablation-warmup}. We can summarize two key takeaways: 1) the warmup steps are indeed useful to make the constrained SFT to be perform consistently safer; 2) the safety improvement is not solely coming from the warmup steps but mostly from the constrained SFT optimization objective we design.

\begin{table}[!htbp]
\caption{Ablation on The Effects of The 10 Warmup Steps. (Fine-tuning Llama-2-7B-Chat)}
\centering
  \resizebox{.95\linewidth}{!}{

\begin{tabular}{c|c|c|c|c|c|c}
\hline
\multirow{2}*{\shortstack{Datasets}} & & \multirow{2}*{\shortstack{Initial}} & \multirow{2}*{\shortstack{Standard\\SFT}} & \multirow{2}*{\shortstack{Standard SFT\\(with warmup)}} & \multirow{2}*{\shortstack{Constrained\\SFT}} & \multirow{2}*{\shortstack{Constrained SFT\\(with warmup)}} \\ 
&   &  &  &  &   & \\
\hline
\multicolumn{7}{c}{\textit{Against Fine-tuning Attacks}}\\
\hline
Harmful Examples & ASR   &  1.5\%  &  88.9\%  &  89.4\%  &  29.1\%  & 4.6\%  \\
\hline
Identity Shifting & ASR   &  0\% & 
 79.5\%   &  44.8\%  &  69.6\%  & 8.1\% \\
%\hline
%Data Selection~\citep{he2024s} &  ASR  & 3.6\%  &    &    &    &  \\
\hline
\multirow{2}*{\shortstack{Backdoor\\Poisoning}}  & ASR (w/o trigger)   & 1.5\%   & 7.6\%   &  2.7\%   & 2.7\%   &  1.9\% \\
 & ASR (w/ trigger)  &  1.7\%  & 90.9\%   & 80.5\%   &  9.7\%  & 10.9\%  \\
\hline
\multicolumn{7}{c}{\textit{Fine-tuning with Normal Downstream Datasets}}\\
\hline
\multirow{2}*{\shortstack{Samsum}} & ASR  & 1.5\% &   23.4\%  &   3.8\% & 23.1\% &  3.2\%  \\
&  Utility & 25.5\% & 51.7\%  &  51.9\% & 50.2\% &  50.1\%  \\
\hline
\multirow{2}*{\shortstack{SQL Create Context}} & ASR  & 1.5\%  & 15.4\%  &  3.3\%  & 2.0\% & 3.2\%   \\
&  Utility & 14.9\% & 99.1\%  & 99.1\% & 98.6\% &  98.5\% \\
\hline
\multirow{2}*{\shortstack{GSM8k}} & ASR  & 1.5\% & 3.3\% & 2.9\%  & 3.1\% &  2.0\%  \\
&  Utility & 25.5\% & 41.7\% &  41.6\% & 37.2\% & 37.4\% \\
\hline
\end{tabular}}
\label{tab:llama-2-finetuning-ablation-warmup}
\end{table}

\textbf{Fine-tuning Attacks on The Augmented Model We Build in Section~\ref{sec:make_alignment_deeper}.} Finally, we also repeat the same set of fine-tuning experiments on the augmented model that we build in Section~\ref{sec:make_alignment_deeper}. Results are presented in Table~\ref{tab:llama-2-augmented-finetuning}. By comparing the results of SFT in Table~\ref{tab:llama-2-augmented-finetuning} and Table~\ref{tab:llama-2-finetuning-ablation-warmup}, we can see the augmented model is generally more robust in multiple fine-tuning cases compared with the non-augmented model. And the constrained fine-tuning objective can also be applied to it, though we didn't observe consistently better results when the two techniques are combined.

\begin{table}[!htbp]
\caption{Fine-tuning Llama-2-7B-Chat-Augmented That We Built in Section~\ref{sec:make_alignment_deeper}. Refer to Table~\ref{tab:llama-2-finetuning-ablation-warmup} for the results on the non-augmented counterpart, i.e., Llama-2-7B-Chat.}
\centering
  \resizebox{.95\linewidth}{!}{

\begin{tabular}{c|c|c|c|c|c}
\hline
\multirow{2}*{\shortstack{Datasets}} & &  \multirow{2}*{\shortstack{Standard\\SFT}} & \multirow{2}*{\shortstack{Standard SFT\\(with warmup)}} & \multirow{2}*{\shortstack{Constrained\\SFT}} & \multirow{2}*{\shortstack{Constrained SFT\\(with warmup)}} \\ 
&   &  &  &  &    \\
\hline
\multicolumn{6}{c}{\textit{Against Fine-tuning Attacks}}\\
\hline
Harmful Examples & ASR   &  55.2\%  &  51.5\%  &  9.4\%  & 5.2\%  \\
\hline
Identity Shifting & ASR   & 
 53.9\%   &  37.0\%  &  28.8\%  & 3.0\% \\
%\hline
%Data Selection~\citep{he2024s} &  ASR  & 3.6\%  &    &    &    &  \\
\hline
\multirow{2}*{\shortstack{Backdoor\\Poisoning}}  & ASR (w/o trigger)   &  3.9\%   &  2.7\%   & 1.8\%   &  0.9\% \\
 & ASR (w/ trigger)  & 80.0\%   & 83.6\%   &  16.4\%  & 12.7\%  \\
\hline
\multicolumn{6}{c}{\textit{Fine-tuning with Normal Downstream Datasets}}\\
\hline
\multirow{2}*{\shortstack{Samsum}} & ASR  & 2.1\%  &   0.6\% & 2.1\% &  1.2\%  \\
&  Utility & 52.4\%  &  51.9\% & 50.4\% &  50.1\%  \\
\hline
\multirow{2}*{\shortstack{SQL Create Context}} & ASR  & 3.8\%  &  1.5\%  & 2.0\% & 0.9\%   \\
&  Utility & 99.0\%  & 99.1\% & 98.4\% &  98.5\% \\
\hline
\multirow{2}*{\shortstack{GSM8k}} & ASR  & 0.9\% & 0.9\%  & 0.3\% &  0.3\%  \\
&  Utility & 42.0\% &  41.2\% & 36.5\% & 36.9\% \\
\hline
\end{tabular}}
\label{tab:llama-2-augmented-finetuning}
\end{table}

%\section{Deviating The Loss Function from An RL Perpsective}
\section{Interpretation of Our Constrained Fine-tuning Objective}
\label{appendix:derivation_of_loss_regularization}

In this section, we provide detailed interpretation for our constrained fine-tuning objective in~\Cref{sec:constraining}.
Recall that our fine-tuning objective is defined as:
\begin{equation}
    \cL(\theta) := \mathop{\mathbb E}_{(\bx,\by)\sim D} -  \sum_{t=1}^{|\by|} \frac{2}{\beta_t} \log\Bigg[\sigma \bigg( \beta_t \log \frac{\pi_{\theta}\big(y_t \mid \bx, \by_{<t} \big)}{\pialigned\big(y_t \mid \bx, \by_{<t} \big)}  \bigg)  \Bigg].
\end{equation}
Alternatively, we can rewrite the fine-tuning objective by linearity of expectation as:
\begin{align}
    \cL(\theta)
    &= \mathop{\mathbb E}_{(\bx,\by)\sim D} \sum_{t \ge 1} -\onec{t \le |\by|}\frac{2}{\beta_t} \log\Bigg[\sigma \bigg( \beta_t \log \frac{\pi_{\theta}\big(y_t \mid \bx, \by_{<t} \big)}{\pialigned\big(y_t \mid \bx, \by_{<t} \big)}  \bigg)  \Bigg] \\
    &= \sum_{t \ge 1} \mathop{\mathbb E}_{(\bx,\by)\sim D}  -\onec{t \le |\by|} \frac{2}{\beta_t} \log\Bigg[\sigma \bigg( \beta_t \log \frac{\pi_{\theta}\big(y_t \mid \bx, \by_{<t} \big)}{\pialigned\big(y_t \mid \bx, \by_{<t} \big)}  \bigg)  \Bigg] \\
    &= \sum_{t \ge 1} \mathop{\mathbb E}_{(\bx,\by)\sim D}  \onec{t \le |\by|} \frac{2}{\beta_t} S\bigg(-\beta_t \log \frac{\pi_{\theta}\big(y_t \mid \bx, \by_{<t} \big)}{\pialigned\big(y_t \mid \bx, \by_{<t} \big)}  \bigg),
\end{align}
where in the last equality we define $S(z) := -\log(\sigma(-z)) = -\log(\frac{1}{1 + \exp(z)}) = \log(1 + e^z)$, namely the softplus function.
Therefore, we can split $\cL(\theta)$ into a sum of token-wise losses:
\begin{align}
    \cL(\theta) = \sum_{t \ge 1} \ell_t(\theta),~\text{where}~&\ell_t(\theta) :=  \mathop{\mathbb E}_{(\bx,\by)\sim D}  \onec{t \le |\by|} \frac{2}{\beta_t} S\bigg( \beta_t \Delta_t(\bx, \by_{<t}, y_t) \bigg), \label{eqn:soft-sft-token-wise-loss-decouping}\\
    & \Delta_t(\bx, \by_{<t}, y_t) := \log \pialigned\big(y_t \mid \bx, \by_{<t} \big) - \log \pi_{\theta}\big(y_t \mid \bx, \by_{<t} \big). \label{eqn:delta-term-in-loss}
\end{align}
In the following, we mainly focus on how to interpret the token-wise loss $\ell_t(\theta)$.

\subsection{Limiting Behaviors}
\label{appendix:loss_limiting}

\subsubsection{Small $\beta_t$} 

When $\beta_t$ is small, we have the following theorem showing that $\ell_t(\theta)$ in Eqn~\ref{eqn:soft-sft-token-wise-loss-decouping} is \textbf{approximately the cross-entropy loss} at position $t$, up to a constant.

\begin{comment}
{\color{red} AHMAD: This needs to be stated rigorously. Here is my suggestion. Let's first define $\ell_{t, \beta}(\theta)$ to make dependence on $t$ and $\beta$ explicit. Then, the claim is that
Let 
\begin{equation}
    F_{t}(\theta) := \frac{1}{2} \E_{(\bx,\by)\sim D} \left[ -\onec{t \le |\by|} \cdot \log \pi_{\theta}(y_t \mid \bx, \by_{<t}) \right]
\end{equation}
Then, 
\begin{equation}
    \lim_{\beta \to 0} \left( \ell_{t, \beta} (\theta) -  C(\beta) \right)= F_{t} (\theta),
\end{equation}
where 
\begin{equation}
    C(\beta) := \E_{(\bx,\by)\sim D} \left[\onec{t \le |\by|} \left( \frac{1}{\beta} \log 2 - \log \pialigned(y_t \mid \bx, \by_{<t}) \right) \right]
\end{equation}
AHMAD: To me this begs the question of why we don't bake $C(\beta)$ into the objective for its behavior to be rational at $\beta = 0.$
}

\xiangyu{To Kaifeng, due to $C(\beta)$, $\ell$ goes infinity and no longer well-defined. Though it's minor.}
\end{comment}

\begin{theorem}
    For a given $\theta$, as $\beta_t \to 0$, we have
    \begin{equation}
        \ell_t(\theta) - C(\beta_t) =  \E_{(\bx,\by)\sim D} \left[ -\onec{t \le |\by|} \cdot \log \pi_{\theta}(y_t \mid \bx, \by_{<t}) \right] + O(\beta_t), 
    \end{equation}
where 
    \begin{equation}
    \quad C(\beta_t) := \E_{(\bx,\by)\sim D} \left[\onec{t \le |\by|} \left( \frac{2}{\beta_t} \log 2 + \log \pialigned(y_t \mid \bx, \by_{<t}) \right) \right],
    \end{equation}
which is a bias term that is constant with respect to $\theta.$
\end{theorem}
\begin{proof}
Recall that $\ell_t(\theta) :=  \mathop{\mathbb E}_{(\bx,\by)\sim D}  \onec{t \le |\by|} \frac{2}{\beta_t} S( \beta_t \Delta_t(\bx, \by_{<t}, y_t))$ and $S(z) := \log(1 + e^{z})$ is the softplus function~\citep{dugas2000incorporating}. By Taylor expansion, it holds for all $z$ that $S(\beta_t z) = S(0) + S'(0) \beta_t z + S''(\xi \beta_t z) \beta_t^2 z^2$ for some $\xi \in (0, 1)$ depending on $z$.
Since $S(0) = \log 2$, $S'(0) = \frac{1}{2}$ and $S''(x) \in [0, 1/4]$ for all $x$,
we have $|S(\beta_t z) - (\log 2 + \frac{1}{2} \beta_t z)| \le \frac{1}{4} \beta_t^2 z^2$.
Dividing both sides by $\beta_t/2$, we get
\begin{align}
    \left|\frac{2}{\beta_t}S(\beta_t z) - \left(\frac{2}{\beta_t}\log 2 +  z \right)\right| \le \frac{1}{2} \beta_t z^2.
\end{align}
Then for our token-wise loss $\ell_t(\theta)$, we have
\begin{align*}
    \left| \ell_t(\theta) 
    - \E_{(\bx,\by)\sim D}  \left[ \onec{t \le |\by|} \left( \frac{2}{\beta_t} \log 2 +  \Delta_t(\bx, \by_{<t}, y_t) \right)\right] \right|
    &\le \frac{1}{2} \beta_t \E_{(\bx,\by)\sim D}  \onec{t \le |\by|} \Delta_t^2 (\bx, \by_{<t}, y_t) \\
    & = O(\beta_t).
\end{align*}
Recall that $\Delta_t(\bx, \by_{<t}, y_t) := \log \pialigned\big(y_t \mid \bx, \by_{<t} \big) - \log \pi_{\theta}\big(y_t \mid \bx, \by_{<t} \big)$ as per Eqn~\ref{eqn:delta-term-in-loss}. We can thus have
\begin{equation}
\begin{aligned}
    \ell_t(\theta) 
    &- \E_{(\bx,\by)\sim D}  \left[ \onec{t \le |\by|} \left( \frac{2}{\beta_t} \log 2 +   \log \pialigned(y_t |  \bx, \by_{< t}) \right)\right] 
     \\ &\qquad\qquad\qquad\qquad
    = - \E_{(\bx,\by)\sim D}  \left[ \onec{t \le |\by|}      \log \pi_{\theta}(y_t |  \bx, \by_{< t}) \right] + O(\beta_t).
\end{aligned}
\end{equation}
Noting that $C(\beta_t) := \E_{(\bx,\by)\sim D} \left[\onec{t \le |\by|} \left( \frac{2}{\beta_t} \log 2 + \log \pialigned(y_t \mid \bx, \by_{<t}) \right) \right]$, we can complete the proof.
\end{proof}
%We can thus write $\ell_t(\theta)$ as
%\begin{align*}
%    \ell_t(\theta) = 
%    \E_{(\bx,\by)\sim D} \left[ \onec{t \le |\by|} \left( \frac{1}{\beta_t} \log 2 + \frac{1}{2} \Delta_t(\bx, \by_{<t}, y_t) \right)\right]
%    + O(\beta_t)
%\end{align*}
%Replacing $\Delta_t(\bx, \by_{<t}, y_t)$ by its definition, we get the desired result.

\subsubsection{Large $\beta_t$} 

%{\color{red} I think this needs to be similarly stated in a more rigorous manner.}
When $\beta_t$ is large, we have the following theorem showing that $\ell_t(\theta)$
can be approximated by $\tilde{\ell}_t(\theta) := \E_{(\bx,\by)\sim D}  \left[ \onec{t \le |\by|} \max\{ \Delta_t(\bx, \by_{<t}, y_t), 0 \} \right]$.
\begin{theorem}
For a given $\theta$, as $\beta_t \to +\infty$, we have
\begin{equation*}
    \ell_t(\theta) = 2 \tilde{\ell}_t(\theta) + O(\beta_t^{-1}),
\end{equation*}
where 
\begin{equation*}
    \tilde{\ell}_t(\theta) := \E_{(\bx,\by)\sim D}  \left[ \onec{t \le |\by|} \max\{ \Delta_t(\bx, \by_{<t}, y_t), 0 \} \right].
\end{equation*}
\end{theorem}
\begin{proof}
    First, we note the following identity.
    \begin{align}
        S(z) = \log(1 + e^z) = \max\{z, 0\} + \log((1+e^z) e^{-\max\{z, 0\}})
        = \max\{z, 0\} + \log(1 + e^{-|z|}).
    \end{align}
    This implies that
    \begin{align}
        \left|\frac{2}{\beta_t}S(\beta_t z) - 2 \cdot \max\{z, 0\}\right|
            = \frac{2}{\beta_t} \log(1 + e^{-\beta_t |z|})
            \le \frac{2}{\beta_t} e^{-\beta_t |z|}.
    \end{align}
    Then for our token-wise loss $\ell_t(\theta)$, we have
    \begin{align}
        \left| \ell_t(\theta) 
        - 2 \cdot \E_{(\bx,\by)\sim D}  \left[ \onec{t \le |\by|} \max\{ \Delta_t(\bx, \by_{<t}, y_t), 0 \} \right]
        \right|
        \le 
        \frac{1}{\beta_t} \E_{(\bx,\by)\sim D}  \left[ \onec{t \le |\by|} e^{-\beta_t |\Delta_t(\bx, \by_{<t}, y_t)|} \right]
    \end{align}
    Noting that the RHS is $O(\beta_t^{-1})$ completes the proof.
\end{proof}
Next, we show that $\tilde{\ell}_t(\theta)$ is minimized to $0$ if and only if $\pi_{\theta}(y_t \mid \bx, \by_{<t}) = \pialigned(y_t \mid \bx, \by_{<t})$.
\begin{theorem} \label{thm:relu_loss_min}
    The minimum value of $\tilde{\ell}_t(\theta)$ is $0$.
    If $\Pr_{(\bx, \by) \sim D}[y_t = c \mid \bx, \by_{<t}] > 0$ for all $(\bx, \by)$ in the support of $D$ and all $c$ in the vocabulary,
    then $\tilde{\ell}_t(\theta) = 0$ if and only if $\pi_{\theta}(y_t \mid \bx, \by_{<t}) = \pialigned(y_t \mid \bx, \by_{<t})$.
\end{theorem}
\begin{proof}
    It is obvious that $\tilde{\ell}_t(\theta) \ge 0$, and $\tilde{\ell}_t(\theta) = 0$ can be attained when $\theta$ stays the same as the parameter for $\pialigned$. It is obvious that  $\pi_{\theta}(y_t \mid \bx, \by_{<t}) = \pialigned(y_t \mid \bx, \by_{<t})$
    implies $\tilde{\ell}_t(\theta) = 0$.
    Conversely, suppose $\tilde{\ell}_t(\theta) = 0$.
    Then $\Delta_t(\bx, \by_{<t}, y_t) \le 0$ for all $(\bx, \by)$ in the support of $D$.
    By definition of $\Delta_t(\bx, \by_{<t}, y_t)$, this implies $\pialigned(y_t \mid \bx, \by_{<t}) \le \pi_{\theta}(y_t \mid \bx, \by_{<t})$.
    Summing over all $y_t$ in the vocabulary, we get $1 = \sum_{y_t} \pialigned(y_t \mid \bx, \by_{<t}) \le \sum_{y_t} \pi_{\theta}(y_t \mid \bx, \by_{<t}) = 1$.
    So all the inequalities must be equalities, and we get $\pi_{\theta}(y_t \mid \bx, \by_{<t}) = \pialigned(y_t \mid \bx, \by_{<t})$.
\end{proof}

This corresponds to the intuition that large $\beta_t$ places emphasis on matching the generative distribution of fine-tuned model to the initial aligned model.

%\subsection{Minimizers}
%\label{appendix:loss_minimizers}
%\kaifeng{TBA}

\subsection{Gradients of The Constrained Fine-tuning Objective}
\label{appendix:loss_gradients}

%\xiangyu{I know it is also already in the main content, but I feel keeping it here is not bad for a complete view when readers read this appendix section. Also I find it helpful to refer readers to this subsection for a recap when I need to explain why a division of $1/\beta_t$ is in our loss function.}

%$S(z) := -\log(\sigma(-z)), S'(z) = -\frac{1}{\sigma(-z)} \sigma'(-z) = -\frac{1}{\sigma(-z)} \cdot \sigma(-z) \cdot (1-\sigma(-z)) \cdot -1 = 1 - \sigma(-z) = \sigma(z)$ 

The gradient of $\ell_t(\theta)$ on a data point $(\bx, \by)$ with $|\by| \ge t$ is derived as:
\begin{align}
    \nabla \left( \frac{2}{\beta_t} S(\beta_t \Delta_t(\bx, \by_{<t}, y_t)) \right)
    &= 
    2 S'(\beta_t \Delta_t(\bx, \by_{<t}, y_t)) (-\nabla \log \pi_{\theta}\big(y_t \mid \bx, \by_{<t})) \nonumber \\
    &= 2  \sigma(\beta_t \Delta_t(\bx, \by_{<t}, y_t)) (-\nabla \log \pi_{\theta}\big(y_t \mid \bx, \by_{<t})).
    %&  -  \sum_{t=1}^{|\by|} w_t \times \nabla \ \log \ \pi_{\theta}\big[y_t|\bx, y_1,...,y_{t-1}\big],\ w_t := \Bigg[1 - \sigma\Bigg( \beta \cdot \log \frac{\pi_{\theta}\big[y_t|\bx, y_1,...,y_{t-1}\big]}{\pi_{aligned}\big[y_t|\bx, y_1,...,y_{t-1}\big]}\Bigg)\Bigg]
\end{align}
%\ahmad{There is now repitition between the content here and the main paper.}

Note that the gradient of the cross-entropy loss is $-\nabla \log \pi_{\theta}\big(y_t \mid \bx, \by_{<t}\big)$.
Therefore, compared to vanilla cross-entropy loss, our fine-tuning objective essentially applies an additional adaptive weight $w_t := 2\cdot \sigma(\beta_t \Delta_t(\bx, \by_{<t}, y_t))$ for the token-wise gradient term of cross-entropy. The weight $w_t$ decreases as $\Delta_t(\bx, \by_{<t}, y_t) := \log \pialigned\big(y_t \mid \bx, \by_{<t} \big) - \log \pi_{\theta}\big(y_t \mid \bx, \by_{<t} \big)$ decreases.
This difference in log probabilities essentially characterizes the deviation between the fine-tuned model $\pi_{\theta}$ and the initially aligned model $\pialigned$ at the token position $t$ (see also~\Cref{thm:relu_loss_min}).
Taking together, we can see that the fine-tuning objective will adaptively diminish the weight of those token positions where the deviation of the distribution approaches a certain threshold (controlled by $\beta_t$), thereby constraining it from further deviation (by diminishing the gradient at this position). Note that, at the beginning of the fine-tuning when $\pi_\theta$ is initialized as $\pialigned$, the weight $w_t = 1$, and the gradient of the loss is equivalent to that of standard cross-entropy loss.%\footnote{This is also why we always multiply a constant factor $\frac{2}{\beta_t}$ to the loss term. With this constant factor, we make sure the gradient magnitude is consistent with the standard SFT (with cross-entropy loss) when the weight $w_t$ does not diminish.}

\subsection{Interpreting Eqn~\ref{eqn:soft-sft-fine-tuning-loss} from A Reinforcement Learning Perspective}
\label{appendix:loss-rl-perspective}

%\begin{align*}
%    \min_{\theta} \Bigg\{ \ \ \mathop{\mathbb E}_{(\bx,\by)\sim D} -  \sum_{t=1}^{|\by|} \frac{2}{\beta_t} \log\Bigg[\sigma \bigg( \beta_t \log \frac{\pi_{\theta}\big(y_t \mid \bx, \by_{<t} \big)}{\pialigned\big(y_t \mid \bx, \by_{<t} \big)}  \bigg)  \Bigg] \ \ \Bigg\},
%\end{align*}

We note that our loss function in Eqn~\ref{eqn:soft-sft-fine-tuning-loss} can also be interpreted from a reinforcement learning perspective if we \textbf{cast fine-tuning as a KL-constrained reinforcement learning problem} rather than a standard supervised fine-tuning problem. 
Specifically, we can follow a similar trick as DPO~\cite{rafailov2023direct} to derive a unified loss, but taking a different approach to reward modeling. 
We, instead, formulate our problem setting like \citet{mudgal2024controlled}, where we optimize at the token level (where tokens are actions), rather than DPO which uses a sequence-level optimization (corresponding more to a bandit setting where entire sequences are actions).

\subsubsection{A Token-wise Reinforcement Learning~(RL) Formulation.}
\label{appendix-subsubsection:token-wise-rl-formulation-basics}

We first introduce the following token-wise RL formulation for LM alignment, which is adapted from \citet{mudgal2024controlled}. In Appendix~\ref{appendix-subsubsec:cast-custom-fine-tuning-to-RL}, we show how a fine-tuning task can be cast into this token-wise RL problem, and Eqn~\ref{eqn:soft-sft-fine-tuning-loss} is essentially a surrogate learning objective of this RL problem.

% \peter{I'm experimenting with some notation changes here, still in flux, and I might revert to the previous version. }

\textbf{Reward Function.} For a pair of an input and a model response $(\bx, \by)$, we cast custom fine-tuning as a problem of further optimizing an already aligned model for a new custom reward function $r([\bx, \by])$. Here we use $[\bx, \by]$ to denote a concatenation of the two sequences and we use this concatenation to denote a state, and the reward function is defined on this state. Following \citet{mudgal2024controlled}, we can decompose it to a token-wise reward $R([\bx, \by_{< t}])$ defined on the intermediate state $[\bx, \by_{< t}]$:
\begin{align}
    & R([\bx, \by_{< t}]) = 
    \begin{cases}
        0 , & y_{t-1} \ne EOS\\
        r([\bx, \by_{< t}]) , & y_{t-1} = EOS 
    \end{cases},
\end{align}
where $EOS$ is the end of the sequence token. The reward is nonzero only if the decoding is complete. We note that, by $r([\bx, \by_{<t}])$ and $R([\bx, \by_{<t}])$, we mean the function $r$ and $R$ are applied on the concatenation of $\bx$ and $\by_{<t}$. Similarly, in the following, we will also use notations such as $R([\bx, \by_{<t}, \bz_{<\tau}])$ and $R([\bx, \by_{<t}, z])$ to denote that we apply $R$ on the concatenation between $[\bx, \by_{<t}]$ and a followup sequence $\bz_{<\tau}$ or just a single token $z$. These concatenations all represent some states of the generation.

\textbf{Value Function.}
%\peter{The notation here is a little confusing. $V^*$ is typically the optimal value function. Is that what the value is supposed to be here? Seems like it shouldn't be.} \ahmad{True, i think it might be better to drop it but first we probably need to define the notation clearly.}
At an intermediate state $[\bx, \by_{< t}]$, the value function of a policy $\pi$ defined on this reward function can be written as:
\begin{align}
    & V_{\pi}([\bx, \by_{< t}]) := \mathop{\mathbb E}_{\bm{z} \sim \pi( \cdot | \bx, \by_{< t})} \Bigg\{ \sum_{\tau \ge 1} R\Big([\bx, \by_{< t}, \bm{z}_{< \tau}]\Big)\Bigg\}
\end{align}
Here $\bz$ is a sequence, and $\bz_{< \tau}$ is empty when $\tau = 1$ as we index tokens in a sequence starting from the index number 1. Also, in the following formulations, we will have multiple different notations $\pi_{\theta}$, $\pialigned$, $\pi^*$ to denote different policies instances, so we will use $V_{\pi_\theta}$, $V_{\pialigned}$, $V_{\pi^*}$ to differently denote their value functions respectively.

\peter{It's a little unusual to define the first part of the advantage as another V, since the first part is supposed to be action dependent (which a value function shouldn't be, though we're in a bit of an odd setting here).} \ahmad{That it true, but in this very simple setting of LLMs, the action is just the next token, so the $q$ function is just simply value of the state with next token appended. Perhaps we can write a couple of sentences to that effect}
\xiangyu{I think it is probably fine to keep it. Basically, the simplicity of our setting is that the reward is only obtained at the terminal state.}
\xiangyu{Also, I think Eqn-27 basically holds independently for each t. That's why we can define the advantage token-wise.}
\ahmad{I wrote a couple of sentences after (25)}

\textbf{Advantage Function.} 
When the custom fine-tuning is modeled by the reward function $R$, we define an expected advantage function of a fine-tuned model $\pi_\theta$ w.r.t. the initially aligned model $\pialigned$:
\begin{equation}
   \hat{A}_{\pi_{\theta}}([\bx, \by_{< t}]) := \mathop{\mathbb{E}}_{z \sim \pi_{\theta}(\cdot | \bx, \by_{< t})} \Bigg\{V_{\pialigned}\Big([\bx, \by_{< t}, z]\Big) - V_{\pialigned}\Big([\bx, \by_{< t}]\Big)\Bigg\},
\end{equation}

Here the advantage function is defined for non-terminal states $[\bx, \by_{< t}]$ where $y_{t-1}$ is not the ending token $EOS$, and $z$ is a single token sampled by $z \sim \pi_\theta(\cdot | \bx, \by_{< t})$. Note that the left-hand term could also be viewed as a $Q$-function with $z$ being the action. In other words, a language model can be viewed as a fully observable Markov decision process (MDP) with state represented by the concatenation of the prompt and the partially decoded response tokens so far and action represented by the next token that is to be decoded.

\textbf{The Reinforcement Learning Objective.} Following \citet{mudgal2024controlled}, we adopt a token-wise RL learning objective:
\begin{align}
    & \max_{\theta} \ \ \mathop{\mathbb E}_{(\bx, \by) \sim D} \  \Bigg\{ \sum_{t\ge 1} \Big[ \tilde{A}_{\pi_{\theta}}([\bx, \by_{< t}]) - \beta_t \cdot \KL\big(\pi_{\theta}(\,\cdot\, | \bx,  \by_{< t})~\big\|~\pialigned(\,\cdot\, | \bx,  \by_{< t}) \big) \Big] \Bigg\},
\label{eqn:tokenwise-rl-objective}
\end{align}
where the advantage function is optimized at each token position $t$ with a token-wise KL regularizer term $\beta_t \cdot \KL\big(\pi_{\theta}(\,\cdot\, | \bx,  \by_{< t})~\big\|~\pialigned(\,\cdot\, | \bx,  \by_{< t}) \big)$ applied for each token position. The strength of regularization at each token position is controlled by $\beta_t$.

\textbf{Closed Form of The Optimal Solution $\pi^*$.} Omitting some intermediate steps for brevity, by Theorem 2.1 of \citet{mudgal2024controlled}, the optimal policy solution $\pi^*$ of Eqn~\ref{eqn:tokenwise-rl-objective} is:
\begin{align}
    & \pi^*(z | \bx, \by_{< t}) = \frac{1}{Z(\bx, \by_{<t})}  \pialigned (z | \bx, \by_{<t}) \cdot e^{V_{\pi^*}([\bx, \by_{<t}, z]) / \beta_t},
\end{align}
where $Z(\bx, \by_{<t})$ is the partition function, $V_{\pi^*}$ is the value function of this optimal policy. We note that this conveniently allows us to re-arrange terms to represent the optimal value function---similar to the steps taken during the derivation of DPO~\cite{rafailov2023direct}:
\begin{align}
    & V_{\pi^*}([\bx, \by_{<t}]) = \beta_t \cdot \log \frac{\pi^*(z|\bx, \by_{<t})}{\pialigned(z|\bx, \by_{<t})} + \beta_t \cdot \log Z(\bx, \by_{<t})
\label{eqn:token-wise-optimal-solution}
\end{align}

\subsubsection{Casting Fine-tuning into a Reinforcement Learning Objective}
\label{appendix-subsubsec:cast-custom-fine-tuning-to-RL}

In the setting of custom fine-tuning we consider in this work, the dataset $D$ is in the form of $D := \{\bx, \by\}$ with only inputs and example outputs, without preference pairs. 
% Therefore, the standard SFT objective~(Eqn~\ref{eqn:sft}) is often the default choice for such custom fine-tuning. However, it is not the only way to learn from the dataset $D$. 
Now, we show an alternative to the standard SFT objective~(Eqn~\ref{eqn:sft}) for learning from this dataset: casting the optimization as a KL-regularized RL objective in Appendix~\ref{appendix-subsubsection:token-wise-rl-formulation-basics}. 
We will show how our token-wise constrained fine-tuning objective in Eqn~\ref{eqn:soft-sft-fine-tuning-loss} can essentially be derived from this token-wise KL-regularized RL problem!

% \textbf{A Token-wise Reward Modeling to Encode Custom Fine-tuning.} 
% PH: it's not really reward modeling, is it? You can't have a reward that's modeling the value function.
Intuitively, fine-tuning $\pialigned$ further on custom data points $(\bx, \by)$ implicitly assumes that this fine-tuning data is as preferable or more preferable than whatever the model would currently output.
To represent this implied preference that the custom example response $\by$ is better than the current responses from $\pialigned$, we can effectively leverage existing preference learning methods.
Recall that, in preference optimization (with both positive and negative examples), the Bradley-Terry model~\cite{bradley1952rank} is used as a model of the underlying reward function. In our setup, we don't have pair-wise preference data, and we only have inputs and positive examples. However, we can define a boolean  $T([\bx, \by_{<t}], y_t)$ to encode the preference: in the fine-tuning task, $y_t$ is a preferable token output following $[\bx, \by_{<t}]$ compared to the responses from the initial aligned model $\pialigned$. So, given an expected random draw from the aligned policy and the draw from the fine-tuned optimal policy, we can define:
% Then, leveraging Eqn~\ref{eqn:token-wise-optimal-solution}:
%we can therefore further define the optimal value function (reward) to reflect this custom preference: %\ahmad{might be good to ground this for a random draw from $\pi_{aligned}$ to explain what it means; basically we are saying the lilelihood that $y_t$ is preferable over a random draw, right?}
\begin{align}
    & \mathop{\mathbb P}\Big(T([\bx, \by_{<t}], y_t)\Big) = \sigma\Bigg( V_{\pi^*}([\bx, \by_{<t}, y_t]) - 
 \mathop{\mathbb E}_{z \sim \pialigned(\cdot |\bx, \by_{<t} ) } V_{\pi^*}([\bx, \by_{<t}, z])\Bigg)
\label{eqn:theoretical-token-wise-reward-model}
\end{align}

Intuitively, this means that---conditioned on the context $[\bx, \by_{\le t}]$---if there is a
continuation token $y_t$ that is higher than the average reward of actions
sampled from the initial aligned model $\pialigned$, it is likely to be an improved or preferred choice in the custom fine-tuning task. 
The higher the margin $V_{\pi^*}([\bx, \by_{<t}, y_t]) - 
 \mathop{\mathbb E}_{z \sim \pialigned(\cdot |\bx, \by_{<t} ) } V_{\pi^*}([\bx, \by_{<t}, z])$ is, the more likely it is an
improvement. 
% Another view of this would be a Bradley-Terry preference function based on value functions, as opposed to reward functions.

% Compared with traditional RLHF and DPO settings, this reward modeling is defined to replace the standard Bradley-Terry model~\cite{bradley1952rank} that is defiend on (positive, negative) preference data pairs.

%on the based aligned model in a reward model first and
%optimize the model to maximize the reward.

% \textbf{Reinforcement Learning.} 
With this function in mind, combined with Eqn~\ref{eqn:token-wise-optimal-solution}, we can leverage a similar derivation to DPO~\cite{rafailov2023direct} to arrive at a constrained fine-tuning objective that is only dependent on the current policy, but is regularized by the original aligned policy. We plug in the closed form of the optimal value function~(Eqn~\ref{eqn:token-wise-optimal-solution}) into the modeling in Eqn~\ref{eqn:theoretical-token-wise-reward-model}, obtaining:
\begin{align}
    \mathop{\mathbb P}\Big(T([\bx, \by_{<t}], y_t)\Big) &= \sigma \Bigg( \beta_t  \log \frac{\pi^*(y_t|\bx, \by_{< t})}{\pialigned(y_t|\bx, \by_{<t})} - \beta_t \hspace{-.03in}\mathop{\mathbb E}_{z \sim \pialigned(\cdot |\bx, \by_{<t} ) } \hspace{-.03in} \log \frac{\pi^*(z|\bx, \by_{<t})}{\pialigned(z|\bx, \by_{<t})} \Bigg) \nonumber \\
    & = \sigma \Bigg( \beta_t  \log \frac{\pi^*(y_t|\bx, \by_{< t})}{\pialigned(y_t|\bx, \by_{<t})} + \beta_t \KL\Big(\pi^*(\cdot | \bx, \by_{<t}) \big\| \pialigned(\cdot | \bx, \by_{<t}) \Big) \Bigg).
\end{align}
Thus, we don't need to explicitly learn the value function, instead it is implicitly encoded by the policy. Then the optimization objective can become:
% \peter{Wait, this whole thing was saying it's not MLE, but now it is? I'm going to comment this part out.}
% Note that, the custom fine-tuning dataset $D$ provides us the learning signals for learning the optimal policy $\pi^*$ here that corresponds to the optimal reward/value $V_{\pi^*}$. \textbf{Specifically, the custom fine-tuning on a custom dataset $D$ comprising of data pairs $(\bx, \by)$ can now be casted into a maximal likelihood estimation of $\mathop{\mathbb P_{\theta}}\Big(T(\bx, \by_{<t}, y_t)\Big)$ on $D$:}

\begin{align}
    & \max_{\theta} \ \  \mathop{\mathbb E}_{(\bx, \by) \sim D} \Bigg\{ \sum_{t \ge 1} \frac{1}{\beta_t}\log \mathop{\mathbb P_{\theta}}\Big(T([\bx, \by_{<t}], y_t)\Big) \Bigg\} \label{eqn:token-wise-reward-modeling-mle-objective},
\end{align}
where: 
\begin{align}
    & \mathop{\mathbb P_{\theta}}\Big(T([\bx, \by_{<t}], y_t)\Big)  := \sigma \Bigg( \hspace{-.03in}\beta_t  \log \frac{\pi_{\theta}(y_t|\bx, \by_{< t})}{\pialigned(y_t|\bx, \by_{<t})} + \beta_t  \KL\Big(\pi^*(\cdot | \bx, \by_{<t}) \big\| \pialigned(\cdot | \bx, \by_{<t}) \Big) \hspace{-.03in}\Bigg),
    \label{eqn:inverse-rl-theta-objective}
\end{align}
and the division of $\beta_t$ in Eqn~\ref{eqn:token-wise-reward-modeling-mle-objective} normalizes the gradient norm at each position $t$ as we will later see in Eqn~\ref{eqn:final-consistency-result} and also clarified in Section~\ref{subsec:token-wise-objective-intro} (and Appendix~\ref{appendix:loss_gradients}).

% Here, the 
% PH: I really don't follow what this means. Just commenting it out. What is the ``taste''?

% implication of the objective in Eqn~\ref{eqn:token-wise-reward-modeling-mle-objective} is that instead of directly optimizing a maximal likelihood estimator on the custom fine-tuning dataset like in Eqn~\ref{eqn:sft}, we encode the ``taste''
% of the desired improvement represented by the custom fine-tuning dataset $D$ into the reward function $r$ and the value function $V$ defined on it. 
% Then, the policy $\pi_\theta$ is learned by optimizing the policy for the reward model~(defined by Eqn~\ref{eqn:theoretical-token-wise-reward-model}) instead. This reward model is further implicitly learned via the inverse reinforcement learning tricks~(i.e., Eqn~\ref{eqn:token-wise-optimal-solution}, \ref{eqn:inverse-rl-theta-objective}), the same as DPO.

% In this case, the optimization problem is the equivalent of optimizing for the optimality gap between the optimal policy's actions in the dataset and the off-policy expected value of 

Note that $\beta_t \cdot \KL\Big(\pi^*(\cdot | \bx, \by_{<t}) \big\| \pialigned(\cdot | \bx, \by_{<t}) \Big)\ge 0$ is a non-negative constant, we have: %\ahmad{I will also think about the narrative here}
\begin{align}
    & \mathop{\mathbb P_{\theta}}\Big(T([\bx, \by_{<t}], y_t)\Big) \ge \sigma \Bigg( \beta_t \cdot \log \frac{\pi_{\theta}(y_t|\bx, \by_{< t})}{\pialigned(y_t|\bx, \by_{<t})} \Bigg)
\end{align}

So, the objective in Eqn~\ref{eqn:token-wise-reward-modeling-mle-objective} can be replaced with a lower bound surrogate objective $L_{\theta}$:
\begin{align}
    & \mathop{\mathbb E}_{(\bx, \by) \sim D} \Bigg\{ \sum_{t \ge 1} \frac{1}{\beta_t}  \log \mathop{\mathbb P_{\theta}}\Big(T([\bx, \by_{<t}], y_t)\Big) \Bigg\} \ge L_{\theta} 
    \label{eqn:final-consistency-result}
\end{align}
where 
\begin{equation*}
    L_{\theta} := \mathop{\mathbb E}_{(\bx, \by) \sim D} \Bigg\{ \sum_{t \ge 1} \frac{1}{\beta_t} \cdot \log \ \sigma \Bigg( \beta_t \cdot \log \frac{\pi_{\theta}(y_t|\bx, \by_{< t})}{\pialigned(y_t|\bx, \by_{<t})} \Bigg) \Bigg\}.
\end{equation*}

Eqn~\ref{eqn:soft-sft-fine-tuning-loss} is then equivalent to $\min_{\theta} -L_\theta \iff  \max_{\theta}  L_{\theta}$. So, optimizing the objective Eqn~\ref{eqn:soft-sft-fine-tuning-loss} is essentially to maximize the lower-bound of the reinforcement learning objective in Eqn~\ref{eqn:token-wise-reward-modeling-mle-objective}.

\end{document}